\documentclass[preprint,12pt,3p]{elsarticle}

\usepackage{amssymb}

\usepackage[T1]{fontenc}
\usepackage{lmodern}
\usepackage{float}
\usepackage{calc}
\usepackage{textcomp}
\usepackage[hidelinks]{hyperref}
\usepackage{amsmath}
\usepackage{amssymb}

\usepackage{graphicx}
\usepackage{picins}
\usepackage{url}
\usepackage{caption}
\usepackage{epstopdf}
\usepackage{algorithm}
\usepackage{algpseudocode}
\usepackage{setspace}
\newtheorem{lemma}{Lemma}
\newtheorem{proof}{Proof}
\newtheorem{theorem}{Theorem}

\journal{Computers \& Security}

\begin{document}

\begin{frontmatter}

\title{Physical Layer Encryption using a Vernam Cipher}

\author[label1,label2]{Yisroel Mirsky\corref{cor1}}
\address[label1]{Georgia Institute of Technology, School of Computer Science, College of Computing}
\address[label2]{Ben-Gurion University, Department of Software and Information Systems Engineering}
\ead{yisroel@post.bgu.ac.il}
\cortext[cor1]{Corresponding author}

\author[label5]{Benjamin Fedidat}
\address[label5]{Jerusalem College of Technology, Department of Communications Engineering}
\ead{fedidat@g.jct.ac.il}

\author[label5]{Yoram Haddad}
\ead{haddad@jct.ac.il}

\begin{abstract}
	Secure communication is a necessity. However, encryption is commonly only applied to the upper layers of the protocol stack. This exposes network information to eavesdroppers, including the channel's type, data rate, protocol, and routing information. This may be solved by encrypting the physical layer, thereby securing all subsequent layers. In order for this method to be practical, the encryption must be quick, preserve bandwidth, and must also deal with the issues of noise mitigation and synchronization.

	In this paper, we present the Vernam Physical Signal Cipher (VPSC): a novel cipher which can encrypt the harmonic composition of any analog waveform. The VPSC accomplished this by applying a modified Vernam cipher to the signal's frequency magnitudes and phases. This approach is fast and preserves the signal's bandwidth. In the paper, we offer methods for noise mitigation and synchronization, and evaluate the VPSC over a noisy wireless channel with multi-path propagation interference. 
\end{abstract}


\end{frontmatter}


\section{Introduction}\label{sec:introduction}

Knowledge is power. It is in the interest of two communicating parties to secure their communication channel to the degree that no information about the channel or the communication is revealed. Today it is common practice to encrypt the payloads of higher level protocols in the communication protocol stack (\textit{Layer 4} and above in the OSI model). This is analogous to encrypting the content of a letter but not the envelope itself. Doing so allows onlookers (eavesdroppers) to see the frame's header information in plaintext and even modify it. This is an even more serious problem for radio transmissions over public domains \cite{physectut}. For instance, by targeting specific bits, an attacker is able to perform energy-efficient jamming \cite{EnergyEfJam}.

Solutions to this problem have been investigated thoroughly \cite{phySecBook,zhou2013physical,nichols2002wireless}. A common solution is to encrypt the data-link layer (\textit{Layer 2}) before it is passed on to the physical layer (\textit{Layer 1}). MACSec \cite{MACSec} is an example of such a protocol for multi-hop wired networks while the most common wireless encryption protocols are WEP and WPA/2. However, encrypting the \textit{Layer 2} bit-stream does not provide secrecy for all information obtainable from the physical channel's characteristics. Examples of this information include traffic statistics, data rates, number of physical channels, service priorities, data size, data packet frequency, baud rate, modulation, and channel bandwidth. Knowledge of this information can be used to infer the transmission equipment, message importance, channel content, and channel capacity.
This information leakage is analogous to writing letters in code and mailing them to a friend. One who sees these letters cannot explicitly determine the content, however the shape (bandwidth), address (protocol) and transmission frequency (bit-rate) can reveal significant information.

Therefore, in order to make a communication channel completely secure, the channel in its entirety should be protected, like a curtain over the entire operation. By extension, it should also be impossible to determine whether the intercepted signal was originally of a digital or analog origin. This level of security can only be achieved by acting at \textit{layer 1} of the OSI protocol stack. For this reason, research on the topic of physical layer security (\textit{Layer 1}) has gained attention over the years \cite{phySecBook}.

In this paper, we propose an application of the Vernam Cipher to analog signals, which we call the Vernam Physical Signal Cipher (VPSC). The VPSC is unique because it encrypts waveform signals on the frequency domain while achieving a high degree of secrecy on \textit{Layer 1}.

There are several notable advantages to working on the frequency domain:
\begin{singlespace}
	\begin{enumerate}
		\item
		\textbf{Complete Information Privacy}: By encrypting the raw signal itself, no information about the channel is exposed. Regardless of waveform, the encrypted signal appears as white noise.
		\item
		\textbf{Bandwidth Preservation}: This aspect is particularly desirable for radio applications where spectrum is a commodity. This is in contrast to performing modulo-based encryption on samples from the temporal plane, since doing so adds energy to all frequencies in the spectrum. In order to preserve the original signal's bandwidth, transformation in the frequency domain is necessary.
		\item
		\textbf{Selective Band Encryption}: This process is similar to a band-pass filter since an entire signal's spectrum can be presented to the VPSC, but only selected frequency bands will be encrypted (regardless of the bands' content).
		\item
		\textbf{Hardware Parallelization}: The signal's spectrum can be split and then encrypted in parallel by independent processors in real-time. This is useful when dealing with very large bands and weak processors. This modularity makes the technology highly scalable to the consumer's needs. When operating directly on the temporal plane, this type of parallelization cannot be achieved when targeting specific frequency bands.
	\end{enumerate}
\end{singlespace}

Although some of these advantages are present in current wireless security channels, the advantages altogether are unavailable in any single one \cite{nichols2002wireless,zhou2013physical,physectut}. The aim of the VPSC is to provide a secure connection (one that does not leak any information, even about the channel itself) between two communicating parties over a single physical link --such as a radio channel or an optical trunk line. The VPSC can also be applied to a multi-hop network if each link is protected separately, and the routing nodes are considered trusted. 

Altogether, this paper has three main contributions:
\begin{singlespace}
	\begin{enumerate}
		\item
		\textbf{A Generic Physical Layer Cipher}: A method for applying the One-Time Pad \cite{shannon49:secrecy} / Vernam Cipher \cite{vernam19:otp} to the frequency domain, enabling all of the advantages listed above.
		\item
		\textbf{Noise Mitigation for Modulo Operation on Waveforms}: In the digital domain (\textit{Layer 2} and above), modulo-based encryption is impacted by noise indirectly. Therefore we propose two noise mitigation techniques which enable modulo-based signal encryption (on both the time and frequency domains), and a method which combines the two into one. These noise mitigation techniques are necessary because modulo operations on analog signals are extremely sensitive to the presence of noise.
		\item
		\textbf{Wireless Simulations \& Source Code}: We evaluate the VPSC's practicality by simulating its application in a realistic wireless channel based on Rayleigh fading. We also provide the Python source code to the simulator of the VPSC for other researchers to reproduce our work.\footnote{URL redacted for anonymity.}
		
	\end{enumerate}
\end{singlespace}

The remainder of the paper is organized as follows: in Section~\ref{sec:relatedwork}, a brief history and related works are presented, in Section~\ref{sec:cryptomodel}, the cryptographic model and the notations needed to describe the VPSC are presented. In Section~\ref{sec:thevpsc}, the VPSC is described and two noise mitigation techniques are proposed.
In Section~\ref{sec:cryptanalysis}, a cryptanalysis is offered. In Section~\ref{sec:sigsync}, a direct encrypter-decrypter synchronization algorithm is proposed. In Section~\ref{sec:performance} we discuss the complexity of the algorithm, and in Section~\ref{sec:simu}, we offer an evaluation of the VPSC through simulations and experimental results.
In Section~\ref{sec:discussion}, we discuss various aspects of the VPSC along with future work.
Finally, in Section~\ref{sec:conclusion}, we summarize our findings and present our conclusion.

\section{Related Work}
\label{sec:relatedwork}
In 1919, Gilbert Vernam patented a XOR-based cipher known as the Vernam Cipher \cite{vernam19:otp}. This cipher works by applying the XOR operation between a message and a secret pseudo-random key. In 1949, Claude Shannon published a historical paper \cite{shannon49:secrecy} in which communication secrecy was studied from the perspective of information theory. He proved the theoretical significance of the Vernam Cipher and proposed the one-time pad (OTP), also known as the Shannon Cipher System (SCS), a cipher capable of perfect secrecy. The OTP is essentially a Vernam Cipher which uses a truly random key.

Later in 1975, Wyner wrote his seminal paper that describes a degraded wiretap channel and provides information-theoretic concepts needed for the domain \cite{wyner1975wire}. Loosely speaking, a wiretap channel (WTC) is where the sender (Alice) transmits a signal to the legitimate receiver (Bob) while an eavesdropper (Eve) intercepts it. However, the signal Eve intercepts is noisier than Bob's, allowing Bob to obtain information that Eve cannot.

We can relate the SCS to the WTC because both can be directly applied to quantized signal samples. However, there is an inherent difference between the SCS and the WTC. The SCS's secrecy is solely based on  information theory, while the WTC's secrecy is based on exploiting the physical traits of communication media \cite{wyner1975wire,CsiszarKorner}. Though the WTC has a trade-off between the underlying channel capacity (secrecy capacity) and the channel's secrecy, its main advantage over the SCS is that it does not necessarily require a secret key. This is desirable in comparison to the SCS, which requires a key with the same length as the message stream. Although this is a technical impracticality, pseudorandom bit generation can be used to create a cryptographically secure key stream \cite{menezes1996handbook,SurveyOnKeyAlgs} (discussed in Section~\ref{sec:cryptanalysis}).

In \cite{WiretapCapwithNoise}, the authors find the secrecy capacity of a shared key WTC in the presence of noise. Later, the authors in \cite{wiretapWithSecKey} generalized \cite{WiretapCapwithNoise} by considering any channel (not necessarily noisy). It can be seen in these works that a secret key WTC without the presence of noise and with maximum secrecy is essentially an SCS. Therefore, in our work, we focus on the SCS since its theory forms the basis for the VPSC.

In \cite{articleLogscram}, the author proposes two methods for encrypting a physical signal: amplitude log masking (ALM) and sample-wise RSA encryption. In ALM, each signal sample in the time-plane is encrypted by taking the logarithm of the sample multiplied by a random value (a key). However, the ALM method does not provide a high degree of secrecy since ALM simply masks (obfuscates) the samples with a key, and attacks such as correlation analysis  \cite{jo10:crackingdsss} can be used to reveal the masked message. Moreover, ALM is very sensitive to noise since errors are exponentially multiplied during the decryption process. This makes ALM impractical to use in noisy channels.

In the RSA method, the RSA \cite{rivest1978method} cipher is applied to each sample from the time-plane. The RSA method requires that the samples be quantized to discrete values. This is necessary in order to perform the power operation of RSA without float-point overflows. As a result, the RSA method is highly susceptible to even the slightest amount of noise. This is because RSA has a non-linear relationship between the cipher-text and the plain-text. As a consequence, every single rounding error in the cryptogram results in a  completely different deciphered value. This behavior is similar to how the output of a hash map is sensitive to changes in its key. This means that the RSA method is not a practical solution for real-world channels. 

Both the ALM an RSA methods use and a significant amount of energy over the entire spectrum. These are undesirable side effects, especially for wireless channels. The VPSC, on the other hand, uses the same amount of bandwidth as the original signal, and a similar amount of energy as well. Furthermore, the VPSC is much more robust to noise and interference, since the modulo operation of the SCS maps noise close to the original sample's value (discussed further in Section~\ref{sec:simu}).

In \cite{audioFFTscramble}, the authors propose frequency component scrambling (FCS) via a Fast Fourier Transform (FFT) to protect audio channels. Their method is to scramble the frequency components (f.c.) of the given signal to obfuscate its contents. However, scrambling does not provide a high level of secrecy. Regardless of the number of f.c.s, it is possible to descramble the signal by analyzing the correlation of the f.c. magnitudes, similarly to what was done in \cite{jo10:crackingdsss}. For example, if a 16-QAM modulation is applied to the carrier frequency $f_c$, FCS is used to encrypt the band surrounding $f_c$, then FCS would simply move the contents of $f_c$ to a neighboring bin (see Fig.~\ref{fig:spectrum}). This is similar to frequency hopping except applied to a much smaller band, and with only one hop. Therefore, FCS is only applicable to the prevention of casual eavesdroppers. 

Fig.~\ref{fig:spectrum} shows the spectrum of an encrypted 16-QAM signal carried on a 1.9MHz wave, using each of the methods. It can be seen that the methods either do not sufficiently secure the channel (FCS), or use the entire spectrum with a large amount of energy (ALM and RSA).

Therefore, to the best of our knowledge, the VPSC is the first physical layer encryption system capable of perfect secrecy that operates directly on the frequency domain, and is also robust to noise.

\begin{figure}
	\centering
	\includegraphics[width=\columnwidth]{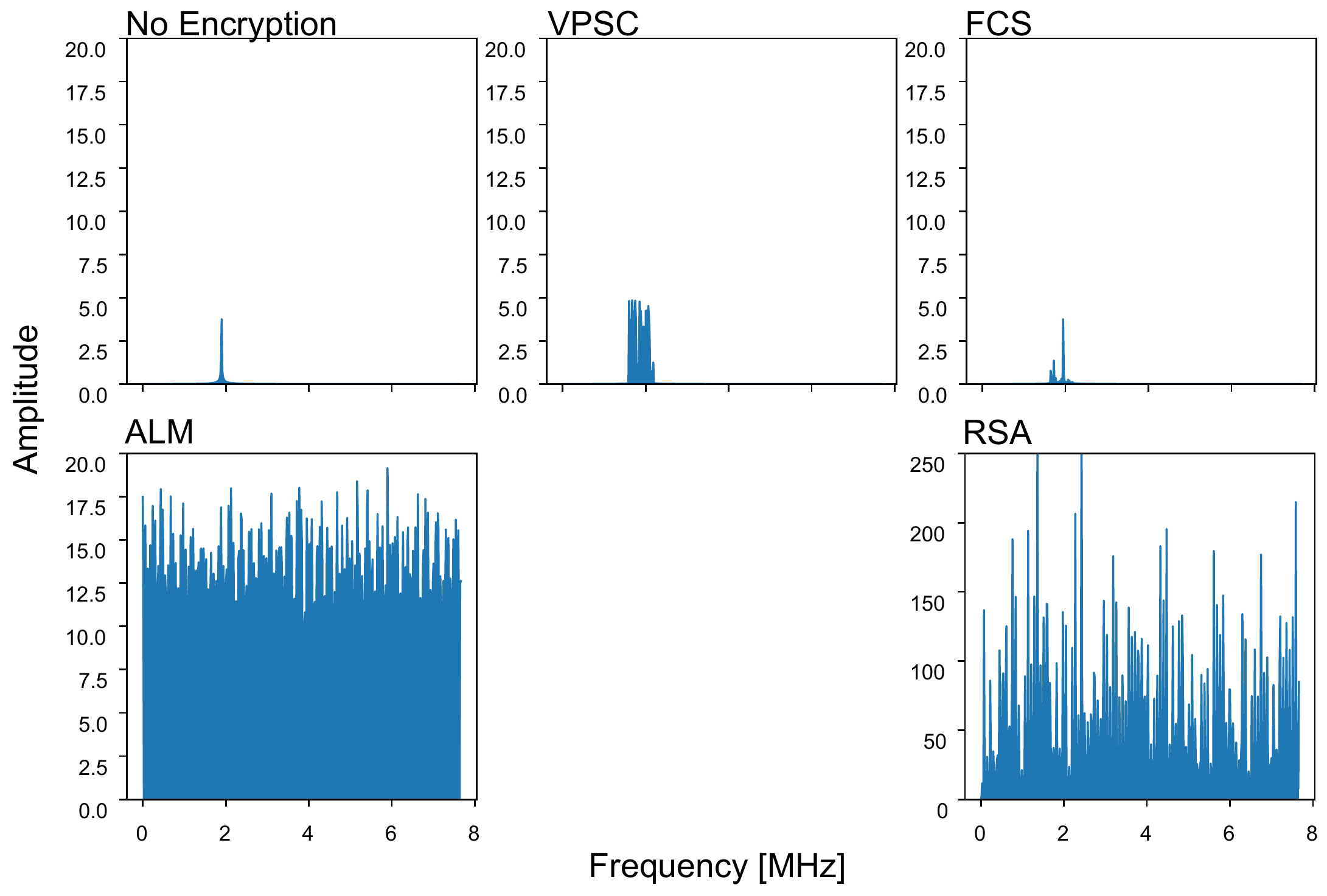} 
	\caption{The spectrum of a 16-QAM channel, with a carrier wave of 1.9 MHz, encrypted by the various methods.}
	\label{fig:spectrum}
\end{figure}

\section{Cryptographic Model}
\label{sec:cryptomodel}
In this Section, we introduce the notations which will be used throughout
this paper. We also apply the standard cryptographic model to the
waveform message space.

In this paper, we assume that we are dealing with discretely sampled
and quantized real signals. Let there be $Q$ discrete levels of
quantization such that the highest level is $Q_U$ and the lowest is $Q_L$.
Let $S \subseteq \mathbb{Z}^N $ be the collection of all possible signal segments
to be encrypted from some particular channel $T$, where $N = 2^k, k \in \mathbb{N}$ and
$\vec{s} \in S$ has the form
\begin{equation}
\vec{s} = \begin{bmatrix} v_1 \\ \vdots \\ v_N \end{bmatrix}
\end{equation}
where $Q_L \leq v_i \leq Q_U$ for every $i$. In other words,
$\vec{s}$ is a frame of $N$ quantized samples, taken from the time domain, which we want to encrypt.

Let $f_s$ be the sample rate of the system such that $f_s \geq 2B$,
where $B$ is the essential bandwidth of the selected signals from $S$ (Nyquist rate).
Let the message space $M$ (the collection of all possible plaintext messages)
of the cryptosystem be defined as the collection of all discrete Fourier
transforms (DFT) of the vectors in $S$, in polar form, such that
\begin{equation}
M = \left\{ \left( \vec{m}_m, \vec{m}_a \right) \middle|
\begin{matrix} \vec{m}_m = | DFT [ \vec{s} ] | \\
\vec{m}_a \angle DFT [ \vec{s} ] \end{matrix}\right\}
\end{equation}
It is helpful to view the message $\vec{m} \in M$ as the polar
form of the DFT of $N$ consecutive samples of a real-time signal segment found in $S$. In other words, $\vec{m}_m$ represents the frequency magnitudes and $\vec{m}_a$ denotes the frequency angles (phases) of the signal segment $\vec{s}$.

Let $\phi$ be a scalar parameter which defines the maximum frequency
magnitude of the cryptosystem. It is restricted to the inequality
\begin{equation}
\label{eq:phi}
\phi \geq max ( \vec{m}_m [ i ] ), \forall i, \vec{m}_m
\end{equation}
\vspace{3pt}
Let the key space $K$ (the collection of all possible keys) of the cryptosystem be defined as a collection of all possible tuples in the form
\begin{equation}
\vec{k} = \left( \vec{k}_m, \vec{k}_a \right)
\end{equation}
where $\vec{k}_m$ and $\vec{k}_a$ are random $N$-length
vectors which are used to encrypt magnitudes and angles respectively. Since we
are dealing with real signals, $\vec{k}_m$ and $\vec{k}_a$ must be structured
to conform to the DFT output from real signals.

\begin{subequations}
	Specifically, $\vec{k}_m$ has the structure
	\begin{equation}
	\label{eq:magkeys}
	\vec{k}_m = concatenate \left( 0, \vec{v}, \texttt{mirror}
	\left( \vec{v} \left[ 2 : \frac{N}{2} \right] \right) \right)
	\end{equation}
	where $\vec{v}$ is a $^N/_2$ length vector of random values
	on the range $[ 0, \phi )$, \texttt{refl} is the mirror rearrangement
	operation on the values of some vector, and the symbol ``:'' indicates a range of
	indexes.
	Similarly, $\vec{k}_a$ has the structure
	\begin{equation}
	\vec{k}_a = concatenate \left( \vec{a}, - \texttt{refl}
	( \vec{a} ) \right)
	\end{equation}
	where $\vec{a}$ is also a $^N/_2$ length vector of random values,
	but on the range $[-\pi,\pi)$.
\end{subequations}

Let the key space $K$ (the collection of all possible keys) of the
cryptosystem be defined as the collection of all possible $\vec{k}$.

The cryptogram space of the system $C$ is equivalent to the collection of all
possible real signals found in $M$. This is necessary in order to obtain perfect
secrecy, since it must be possible to map any cryptogram $\vec{c} \in C$
back to any message $\vec{m} \in M$ \cite{shannon49:secrecy}. 

\begin{singlespace}
	Let the inverse-$DFT$ ($DFT^{-1}$) of the cryptogram $\vec{c}$ be
	referred to as $\vec{s}'$, such that
	\begin{equation}
	\label{eq:reversedft}
	DFT^{-1} ( \vec{c} ) = \vec{s}'
	\end{equation}
	
	\begin{subequations}
		Let the general encryption function be defined as
		\begin{equation}
		\label{eq:encrypt}
		e_{\vec{k}}(\vec{m}) = \vec{c}
		\end{equation}
		and the general decryption function be defined as
		\begin{equation}
		\label{eq:decrypt}
		d_{\vec{k}}(\vec{c}) = \vec{m}
		\end{equation}
		where the key $\vec{k} \in K$ is used to encrypt the message
		$\vec{m} \in M$ and decrypt the ciphertext $\vec{c} \in C$.
	\end{subequations}
\end{singlespace}
Now that some notation has been defined, it is possible to present the
cryptographic model used in this paper, as depicted in
Fig.~\ref{fig:aliceandbob}.
\begin{figure}
	\centering
	\includegraphics[scale=0.35]{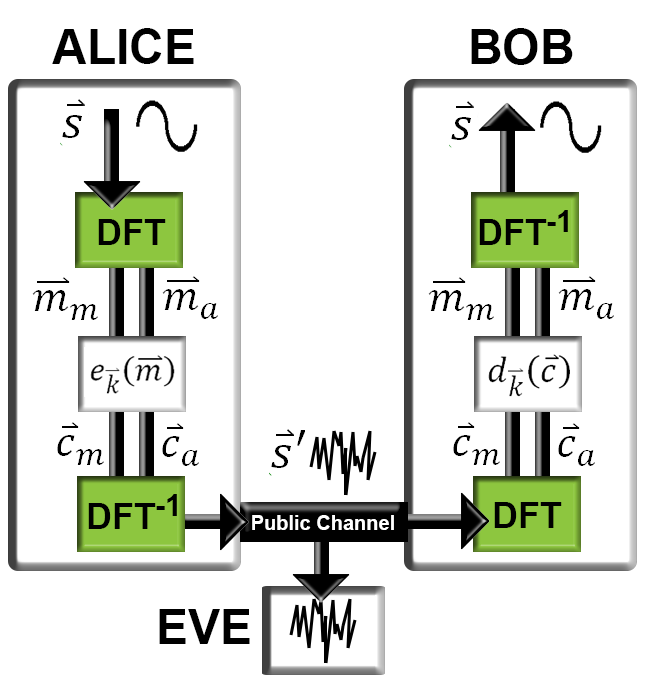} 
	\caption{The cryptosystem model for analog signals. Alice sends a waveform
		signal to Bob, and Eve's interception provides her with no information about
		it.}
	\label{fig:aliceandbob}
\end{figure}
Consider a case in which
Alice wants to transmit $\vec{s}$ (a segment of some analog signal)
securely to Bob so that Eve cannot obtain any information about $\vec{s}$
as it travels across the public medium. First, Alice obtains the tuple
$\vec{m} = (\vec{m}_m, \vec{m}_a )$ by converting
the $DFT$ of $\vec{s}$ into polar form. Next, Alice encrypts the
frequency components ($\vec{m}$) with~\eqref{eq:encrypt} by performing
$e_{\vec{k}} ( \vec{m} )= \vec{c}$. Finally,
Alice performs a $DFT^{-1}$ on the cryptogram $\vec{c}$ and transmits
the result $\vec{s}'$ over the public medium towards Bob.

Once Bob has received $\vec{s}'$ he can obtain the original $\vec{s}$ from it
by performing the same steps which Alice performed, while using the decryption
function~\eqref{eq:decrypt} instead.

This process is repeated continuously in real-time for each set $N$ samples that Alice wishes to send. However, each key $\vec{k}$ is selected at random from $K$ (key generation is further described in the next Section).

\section{The Vernam Physical Signal Cipher}
\label{sec:thevpsc}
In this Section, we define the VPSC by detailing its implementation
of the encryption and decryption functions
(\ref{eq:encrypt},~\ref{eq:decrypt}). We also introduce two methods of
noise mitigation which are essential for the VPSC to work in the real world. Afterwards, we review possible key-sharing options.

\begin{subequations}
	Let the VPSC encryption function be defined as
	\begin{align}
	\label{eq:vpscencrypt}
	\begin{split}
	e_{\vec{k}}(\vec{m}) =
	\left\{ \substack{e_{\vec{k}_m} \left(\vec{m}_m\right) + \lambda \\
		e_{\vec{k}_a} \left(\vec{m}_a\right) } \right\} =
	\left\{ \substack{ \left(\left( \vec{m}_m + \vec{k}_m \right)
		\bmod \phi \right)+\lambda \\ \vec{m}_a + \vec{k}_a } \right\} =
	\left\{ \substack{ \vec{c}_m + \lambda \\ \vec{c}_a } \right\} =
	\vec{c}
	\end{split}
	\end{align}
	and let the VPSC decryption function be defined as
	\begin{align}
	\label{eq:vpscdecrypt}
	\begin{split}
	d_{\vec{k}}(\vec{c}) =
	\left\{ \substack{d_{\vec{k}_m} \left(\vec{c}_m\right)-\lambda \\
		d_{\vec{k}_a} \left(\vec{c}_a\right)} \right\} =
	\left\{ \substack{ \left(\left( \vec{c}_m - \vec{k}_m \right)
		\bmod \phi \right) - \lambda\\ \vec{c}_a - \vec{k}_a } \right\} =
	\left\{ \substack{ \vec{m}_m  - \lambda\\ \vec{m}_a } \right\} =
	\vec{m}
	\end{split}
	\end{align}
	where $\bmod$ is the element-wise modulo operation,
	$\vec{k}$ is a purely random vector selected from $K$, and $\lambda$ is a required amplification of the encrypted signal. The parameter $\lambda$ is a constant defined by the user such that $\lambda > 0$. The purpose of $\lambda$ is to ensure that the phase of component $n$ will not be lost, in the chance that $\vec{c}_m[n] \approx 0$. This can legitimately occur at random, based on $\vec{k}_m$. We note that $\lambda$ does not affect the secrecy of $\vec{c}$ because we are simply amplifying the final signal.
\end{subequations}

\subsection{Noise Mitigation}
Since the VPSC is a physical signal cipher, it must operate according
to physical constraints. One of those is $\phi$; the maximum frequency magnitude of the cryptosystem. This parameter must be at least
as large as the largest possible frequency magnitude in $M$ as described in
\eqref{eq:phi}. Use of a value for $\phi$ which is less than the largest
magnitude will result in a loss of information due to the modulo operation.

The functions (\ref{eq:vpscencrypt},~\ref{eq:vpscdecrypt}) can provide a high level of security since they are essentially an OTP (discussed later in Section~\ref{sec:cryptanalysis}). However, their
implementation in reality (as-is) does not function. This is because every communication medium adds some noise to the channel, whether it is natural noise or some other signal interference. Therefore, under normal circumstances, some energy always gets added or subtracted from some of the frequency magnitudes in $\vec{c}_m$.
This incurs an undesirable effect in the decryption process. 

Depending on
the amount of energy, the subtraction and then modulo of the cryptogram in
\eqref{eq:vpscdecrypt} can send values in $\vec{m}_m$ that were close to
$0$ or $\phi$ to the opposite extreme.

\begin{subequations}
	To illustrate this issue we can track the usage of the VPSC over some noisy channel. Let's say that $\vec{m}_m [n]$ is the
	n\textsuperscript{th} frequency magnitude from the original signal segment $\vec{s}$,
	and that $\vec{m}_m [n] = \phi - \varepsilon$, where
	$\varepsilon$ is some relatively small number. Suppose that the encryption
	key to be used on the magnitude $\vec{m}_m [n]$ is
	$\vec{k}_m [n] = \alpha$, where $0 \leq \alpha < \phi$.
	Encrypting $\vec{m}_m$ with~\eqref{eq:vpscencrypt} results in:
	\begin{equation}
	\vec{c}_m [n]
	= \left( \vec{m}_m [n] + \vec{k}_m [n] \right) \bmod \phi
	= (\phi - \varepsilon + \alpha) \bmod \phi
	\end{equation}
	Now $\vec{c}_m$ is converted into $\vec{s}'$ by
	\eqref{eq:reversedft} and transmitted over the communication medium. Assume that by doing so, the magnitude $\vec{c}_m [n]$ (of
	$\vec{s}'$) receives some additional energy $\gamma$ from noise
	in the channel, where $\gamma > \varepsilon$ . The result is that Bob
	now receives a noisy cryptogram,
	\begin{equation}
	\label{eq:noisyc}
	\vec{c}_m[n]^* = (\phi - \varepsilon + \alpha) \bmod \phi + \gamma
	\end{equation}
	and Bob cannot analytically determine $\gamma$ since the original message
	magnitude $\vec{m}_m [n]=\phi - \varepsilon$ is unknown to him. When Bob tries to decrypt~\eqref{eq:noisyc} using~\eqref{eq:vpscdecrypt},
	assuming $\gamma - \varepsilon < \phi$ the following will occur:
	\begin{equation}  \label{eq:anomalypt1}
	d_{\vec{k}_m[n]} \left( \vec{c}_m [n]^* \right)
	= \left( \vec{c}_m[n]^* - \vec{k}_m[n] \right) \bmod \phi
	\end{equation}
	
	If $\alpha < \gamma$ then~\eqref{eq:anomalypt1} is evaluated to
	\begin{align}
	\label{eq:anomalycase1}
	(\phi - \varepsilon + \alpha + \gamma - \alpha) \bmod \phi = (\phi - \varepsilon + \gamma) \bmod \phi = \gamma - \varepsilon \text{ since } \gamma > \varepsilon
	\end{align}
	
	If $\alpha \geq \gamma$ then~\eqref{eq:anomalypt1} evaluated to
	\begin{equation} \label{eq:anomalycase2}
	(\alpha - \varepsilon + \gamma - \alpha) \bmod \phi = (\gamma - \varepsilon) \bmod \phi
	= \gamma - \varepsilon.
	\end{equation}
	
	In both cases (\ref{eq:anomalycase1},~\ref{eq:anomalycase2}), Bob will
	interpret $\vec{s}$'s n\textsuperscript{th} frequency magnitude to be a near zero value
	as opposed to the correct near maximum value ($\phi$). Similarly, the same issue can
	be shown for near-zero values being interpreted as maximum values as well.
	
	These unavoidable errors add a tremendous amount of noise to the decrypted
	signal. Therefore, since a small amount of noise energy $\gamma$ can cause a
	large signal to noise ratio (SNR), it is impractical to implement the VPSC as-is by
	simply using the encryption and decryption functions (\ref{eq:vpscencrypt},
	\ref{eq:vpscdecrypt}) without any noise mitigation.
	Therefore, we propose two methods of noise mitigation for the VPSC:
	preemptive-rise and statistical-floor.
\end{subequations}

\subsubsection{Preemptive-rise}
The preemptive-rise (PR) technique is implemented both in the encrypter
(transmitter) and decrypter (receiver). The idea is to make a buffer zone
above and below the original signal's range of frequency magnitudes. This ensures that the addition
of random noise will not cause any of the magnitudes to fall out of bounds
during the subtraction step of~\eqref{eq:vpscdecrypt} as depicted in
Fig.~\ref{fig:prmethod}. This is not a conventional signal boost since non-relevant frequencies within the encrypted band will be boosted as well.

\begin{figure}[h]
	\centering
	\includegraphics[width=0.8\columnwidth]{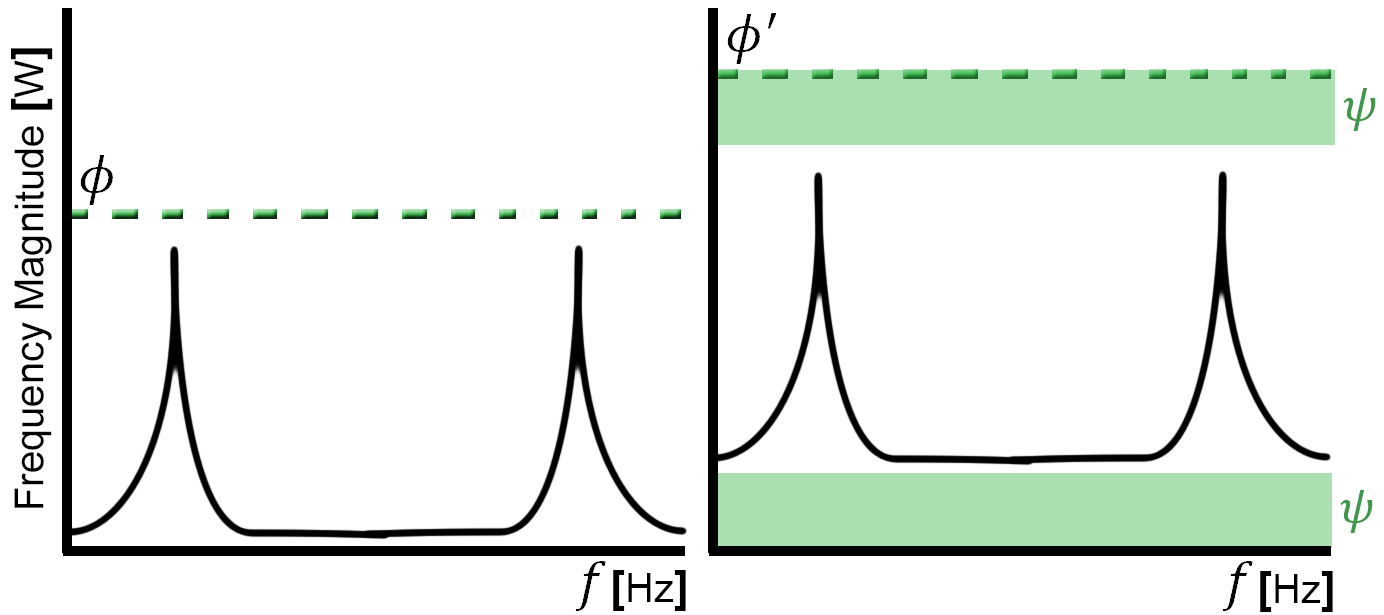}
	\caption{A sine wave signal undergoing the PR technique. The image on the
		right shows the new signal after the procedure where the shaded
		areas provide a modulo-error ``buffer zone'' with a width of $\psi$ watts
		each.}
	\label{fig:prmethod}
\end{figure}

Let $\psi$ be the width of each buffer zone in watts where
\begin{equation}
\psi \equiv u \times \sigma_0
\end{equation}
such that $u \in \mathbb{N}$ and $\sigma_0$ is the standard deviation of the
channel's noise energy.

In order to implement PR, the encrypter and decrypter must use a larger $\phi$
than previously required due to the larger range of magnitude values.
This larger version $\phi'$ can be defined as
\begin{equation}
\label{eq:phitag}
\phi' = \phi + 2 \psi
\end{equation}
where $\phi$ is determined from~\eqref{eq:phi}. Furthermore, the range from which
the magnitude keys can be selected~\eqref{eq:magkeys} must be changed to
$[0, \phi')$.

The implementation of PR is equivalent to modifying the magnitude encryption
function in~\eqref{eq:vpscencrypt} to
\begin{equation}
\label{eq:prencrypt}
e_{\vec{k}_m} \left(\vec{m}_m\right)
= \left( \vec{m}_m + \vec{k}_m + \psi \right) \bmod \phi'
\end{equation}
and magnitude decryption function in~\eqref{eq:vpscdecrypt} to
\begin{equation}
\label{eq:prdecrypt}
d_{\vec{k}_m} \left(\vec{c}_m\right)
= \left( \vec{c}_m - \vec{k}_m \right) \bmod \phi' - \psi
\end{equation}

Although PR can completely eliminate the noise distortions, its cryptogram
\eqref{eq:prdecrypt} requires a greater transmission power than the original
cryptogram~\eqref{eq:vpscencrypt} due to the power change in~\eqref{eq:phitag}.

\subsubsection{Statistical-floor}
Unlike the PR technique, statistical-floor (SF) is implemented in the decrypter alone. The idea is to try and correct those values which have erroneously been shifted over the boundaries by the noise energy. The method tries to clean the signal by correcting erroneous fallouts before and after the subtraction step in~\eqref{eq:vpscdecrypt}.

There are two cases which we consider erroneous: impossible magnitudes and unlikely fallouts.

When there is no added noise to the signal, it is impossible to receive an encrypted magnitude above a certain value. Therefore, when we receive these ``impossible magnitudes'' the only conclusion we can have is that they were affected by some positive noise. More specifically, when there is no noise, the largest frequency magnitude possible is $\phi-\varepsilon$. When there is noise, it is possible to receive a magnitude above $\phi$. Therefore, we can conclude that any received magnitudes greater than or equal to $\phi$ should be floored to $\phi-\varepsilon$ before the subtraction step~\eqref{eq:vpscdecrypt}. An illustration of this technique can be found at the top of Fig.~\ref{fig:sfmethod}.

\begin{figure}[h]
	\center
	\includegraphics[width=\textwidth]{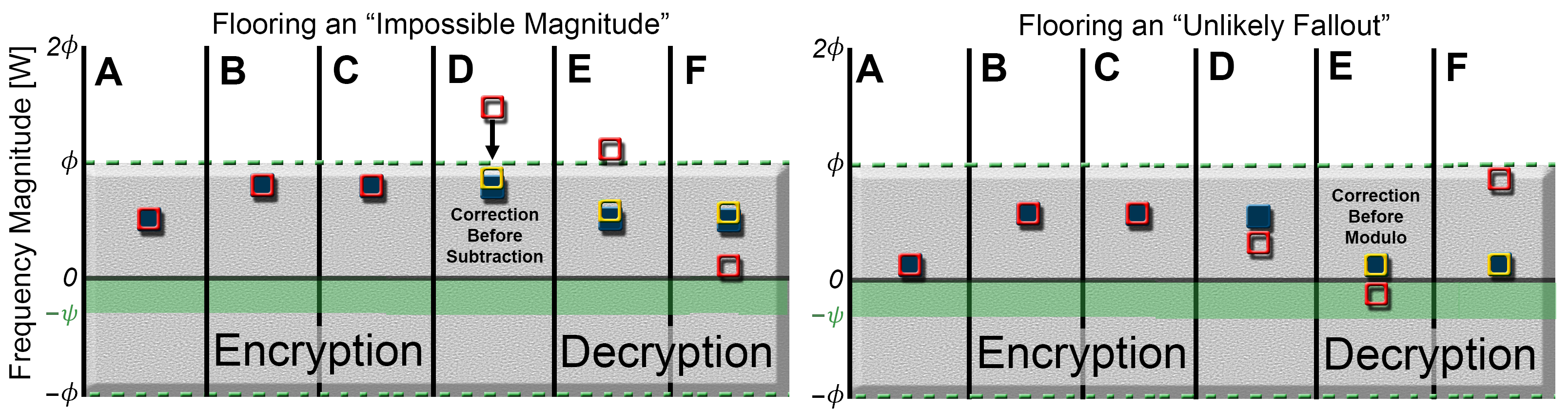}  
	\caption{A single frequency magnitude undergoing SF. Each step of the VPSC encryption/decryption is shown for when the system has no noise (filled square), noise but no SF (red/darker frame), noise and SF (yellow/brighter frame). Steps A-F are as follows: Original value, key added, modulo step, transmission over channel, key removed, modulo step.}
	\label{fig:sfmethod}
\end{figure}

This flooring procedure can be done without any prior knowledge about the original signal. However, after the subtraction step, some values end up just above or below 0. Without knowing the original signal, it is impossible to determine if the values had initially been just above 0 or had been altered as a result of noise. As discussed earlier, these values could add a large amount of noise to the signal after the modulo-step. Fortunately, if we have some statistical information about the original signal, then we can attempt to correct these values (unlikely fallouts).

For instance, let us assume that it is known that the original signal has most of its magnitude values close to 0 (like the signal in Fig.~\ref{fig:prmethod}). In this case, after the subtraction step, we will assume that all values in the range $[-\psi, 0)$ were supposed to be just above 0 but were shifted down by noise. To correct them, we will floor them to 0. The same idea can be applied if we have prior knowledge that the original signal is mostly made up of large magnitudes. An illustration of this technique can be found at the bottom of Fig.~\ref{fig:sfmethod}. This figure shows that the error is minimized by SF in both cases. Furthermore, in the case of ``impossible magnitudes'', a potentially high level of noise can be mitigated.

In Algorithm~\ref{alg:statfloor} the pseudo-code for the SF magnitude decrypter modified from~\eqref{eq:vpscdecrypt}, where $\varepsilon$ is a
very small number, and we assume that the original signal has mostly small magnitudes.

Although SF does not require more power to transmit a cryptogram (unlike PR), it sometimes incorrectly floors values that happen to be below 0 that should not have been modified. Since this error is unavoidable when using this technique, it is impossible for SF to completely eliminate the modulo noise distortions without adding some of its own.

\begin{singlespacing}
	\begin{algorithm}[ht]
		\caption{Pseudo-code for the magnitude decryption function using the SF technique, under the assumption that the original signal has mostly small frequency magnitudes.}
		\label{alg:statfloor}
		\begin{algorithmic}[1]
			\Function{statistical-floor}{$\protect\vec{m}_m$, $\protect\vec{k}_m$}
			\For{$i \gets 1$ to $N$}
			\If{$\vec{m}_m[i] \geq \phi$}\Comment{impossible magnitude}
			\State $\vec{m}_m[i] \gets \phi - \varepsilon$
			\EndIf
			\State $\vec{m}_m[i] \gets \vec{c}_m[i] -
			\vec{k}_m[i]$
			\If{$-\psi \leq \vec{m}_m[i] < 0$}\Comment{unlikely fallout}
			\State $\vec{m}_m[i] \gets 0$
			\EndIf
			\State $\vec{m}_m[i] \gets \left(
			\vec{m}_m[i] \right) \bmod \phi$
			\EndFor
			\EndFunction
		\end{algorithmic}
	\end{algorithm}
\end{singlespacing}	
\vspace{-6pt}

\subsubsection{Combined Method}
It is clear that the PR and SF techniques have their advantages and disadvantages, as mentioned above. To obtain the \textit{best of both worlds}, one may use a combination of both PR and SF techniques. Such a combination can eliminate almost all distortions while using less transmission power than PR to achieve the same noise reduction. The pseudo-code for performing encryption and decryption with the combined method can be found in Algorithms~\ref{alg:enc_combined} and~\ref{alg:dec_combined}.

\begin{singlespacing}
	\begin{algorithm}[ht]
		\caption{Pseudo-code for the VPSC's magnitude encryption, using the combined noise mitigation technique.}
		\label{alg:enc_combined}
		\begin{algorithmic}[1]
			\Function{$e_{\vec{k}_m}$}{$\protect\vec{m}_m$}
			\State $\vec{c}_m \leftarrow \vec{m}_m$
			\For{$i \gets 1$ to $N$}
			\State $\vec{c}_m[i] \leftarrow \vec{c}_m[i] + \lambda$ \Comment{add energy buffer}
			\State $\vec{c}_m[i] \leftarrow \bmod(\vec{m}_m[i] + \vec{k}_m[i], \phi + 2\lambda)$\Comment{encrypt}
			\State $\vec{c}_m[i] \leftarrow \vec{c}_m[i] + \lambda$ \Comment{add carrier energy}
			\EndFor
			\EndFunction
		\end{algorithmic}
	\end{algorithm}
\end{singlespacing}	
\begin{singlespacing}
	\begin{algorithm}[ht]
		\caption{Pseudo-code for the VPSC's magnitude decryption, using the combined noise mitigation technique.}
		\label{alg:dec_combined}
		\begin{algorithmic}[1]
			\Function{$d_{\vec{k}_m}$}{$\protect\vec{c}_m$}
			\State $\vec{m}_m \leftarrow \vec{c}_m$
			\For{$i \gets 1$ to $N$}
			\If{$\vec{c}_m[i] > \phi + 3\lambda$}\Comment{remove impossible magnitudes}
			\State $\vec{c}_m[i] \leftarrow \phi + 3\lambda$
			\EndIf
			\State $\vec{c}_m[i] \leftarrow \vec{c}_m[i] - \lambda$\Comment{remove carrier energy}
			\If{$\vec{c}_m[i] < 0$}
			\State $\vec{c}_m[i] \leftarrow 0$
			\EndIf
			\State $\vec{m}_m[i] \leftarrow \bmod(\vec{c}_m[i] -\vec{k}_m[i], \phi + 2\lambda)$\Comment{decrypt}
			\State $\vec{m}_m[i] \leftarrow \vec{m}_m[i] - \lambda$ \Comment{remove energy buffer}
			\EndFor
			\EndFunction
		\end{algorithmic}
	\end{algorithm}
\end{singlespacing}

\subsection{Key Sharing}
Since the VPSC is a specific case of the OTP, the length of the VPSC's key must equal the length of the streamed message. Sharing this key as a prior secret between two parties is impractical. In this type of situation it is common for both parties to agree upon a secret seed to initialize a key-stream generator. This makes the shared key finite as opposed to virtually infinite in length.

Another consideration is how two communicating parties with only public channels can share a secret key (or seed) before any secure channel has been established. This subject has been well researched \cite{symKeyEx2} and there are several common solutions.  One solution is to form a hybrid cryptosystem using public-key cryptography \cite{symKeyEx1}. In hybrid cryptosystems, a symmetric-key system (such as the VPSC) is initiated with an asymmetric key exchange. This exchange between VPSC transmitter / receiver modules can be controlled by a well-known protocol, and transmitted according to the physical media present.

\section{Cryptanalysis}
\label{sec:cryptanalysis}
In this Section we establish that the VPSC is unconditionally secure
(unbreakable) when using truly random keys. We shall also discuss how the
VPSC is computationally-secure based on the strength of the pseudo random
number generator (PRNG) it uses.

\subsection{Extension of the One-Time-Pad}
The VPSC is based on the Vernam Cipher \cite{vernam19:otp}. Claude Shannon
proved that if truly random numbers are used to generate this cipher's
encryption keys, it then becomes what is known as a One-Time-Pad (OTP), which is
unconditionally secure \cite{shannon49:secrecy}.

The typical OTP operates by performing bitwise XOR operations on two
binary vectors of equal length: the message and the key. The cryptographic behavior
of the XOR operation, performed on each bit in the vector, can be applied to
each element in an n-ary vector as well. This is because the XOR operation can
be viewed as a modulo-2 operation. For instance, if $a,b \in \{0,1\}$ then the
XOR operation $a \oplus b$ is equivalent to performing $(a+b)\bmod 2$.
Similarly, with an n-ary vector, if $c,d \in \{0,\hdots,n-1\}$ then the same
cryptographic behavior is observable from $(a+b)\bmod n$.

This means that the OTP cipher can be applied to any system which uses n-ary
vectors such as $\vec{m}_m, \vec{k}_m, \vec{c}_m$
and $\vec{m}_a, \vec{k}_a, \vec{c}_a$ from
our cryptographic model described in Section~\ref{sec:cryptomodel}. Therefore, the
VPSC can be viewed as an extension of the OTP, and if the selection of
$\vec{k} = (\vec{k}_m, \vec{k}_a) \in K$ is purely random then the VPSC is unconditionally secure.

\begin{lemma}
	\label{magequivocation}
	If the selection of $\vec{k}_m$ is purely random, then 
	\begin{equation}
	\label{eq:otppt1mag}
	p \left( \vec{c}_m [i] = z \middle| \vec{m}_m [i] = x \right)
	= p \left( \vec{c}_m [i] = z \right),
	\forall x \in M_m, \forall z \in C_m, \forall i
	\end{equation}
\end{lemma}

\begin{proof}
	Starting from the left-hand part in~\eqref{eq:otppt1mag}, by Bayes' formula we get
	\begin{equation}
	\label{eq:otppt2}
	p \left( \vec{m}_m [i] = x \middle| \vec{c}_m [i] = z \right)
	= \frac{p \left( \vec{m}_m [i] = x \land \vec{c}_m [i] = z \right)}
	{p \left( \vec{c}_m [i] = z \right)},
	\forall x \in M_m, \forall z \in C_m
	\end{equation}
	The encryption function in the numerator in~\eqref{eq:otppt2} can be expanded to produce the following equivalence
	\begin{subequations}
		\begin{equation}
		\label{eq:numeratorpt1}
		p \left( \vec{m}_m [i] = x \land \vec{c}_m [i] = z \right)
		=  p \left( \vec{m}_m [i] = x \land \vec{k}_m [i]
		= (z-x) \bmod \phi \right)
		\end{equation}
		Since the \emph{a priori} probability of the key selection is independent from
		that of the plaintext magnitudes, the right-hand part of~\eqref{eq:numeratorpt1} can be expressed as
		\begin{equation}
		\label{eq:numeratorpt2}
		p \left( \vec{m}_m [i] = x \right) \times \left(
		\vec{k}_m [i] = (z-x) \bmod \phi \right)
		\end{equation}
		and since the selection of $k_m [i]$ is chosen uniformly on the range
		$[0,\phi]$,~\eqref{eq:numeratorpt2} can be reduced to
		\begin{equation}
		p \left( \vec{m}_m [i] = x \right) \times \frac{1}{L}
		\end{equation}
		where $L$ is the discrete number of levels in the range $[0,\phi]$.
	\end{subequations}
	
	The following equivalence can be applied to the denominator in
	\eqref{eq:otppt2}
	\begin{equation}
	\label{eq:numeratorpt3}
	p \left( \vec{c}_m [i] = z \right) = \sum_X p \left( \vec{m}_m [i] = x \land \vec{c}_m [i] = z \right) = \sum_X p \left( \vec{m}_m [i] = x \right) \times \frac{1}{L} \nonumber = \frac{1}{L} \nonumber
	\end{equation}
	
	In other words, we can deduce from the denominator in~\eqref{eq:otppt2} that
	each cryptogram magnitude $z$ is equally likely to occur. Therefore, from Bayes' theorem
	in~\eqref{eq:otppt1mag} it can be shown that
	\begin{equation}
	\label{eq:magresult}
	p \left( \vec{m}_m [i] = x \middle| \vec{c}_m [i] = z \right)
	= \frac{p \left( \vec{m}_m [i] = x \right) \times \frac{1}{L}}
	{\frac{1}{L}}
	= p \left( \vec{m}_m [i] = x \right),
	\forall x \in M_m, \forall z \in C_m, \forall i 
	\end{equation}
\end{proof}

\begin{lemma}
	\label{angequivocation}
	If the selection of $\vec{k}_a$ is purely random, then 
	\begin{equation}
	\label{eq:otppt1angle1}
	p \left( \vec{m}_a [i] = x \middle| \vec{c}_a [i] = z \right)
	= p \left( \vec{m}_a [i] = x \right),
	\forall x \in M_a, \forall z \in C_a, \forall i 
	\end{equation}
\end{lemma}
\begin{proof}
	Following~\eqref{eq:otppt2}, similarly to~\eqref{eq:numeratorpt1},~\eqref{eq:otppt1angle1} can be presented as
	\begin{equation}
	\label{eq:numeratorpt1angle}
	p \left( \vec{m}_a [i] = x \land \vec{c}_a [i] = z \right)
	=  p \left( \vec{m}_a [i] = x \land \vec{k}_a [i] = (z-x) \right)
	\end{equation}
	
	Since the selection of $k_a [i]$ is chosen uniformly on the range
	$[-\pi,\pi)$, and following ~\eqref{eq:numeratorpt2},~\eqref{eq:numeratorpt1angle} can be reduced to
	\begin{equation}
	p \left( \vec{m}_a [i] = x \right) \times \frac{1}{L}
	\end{equation}
	where $L$ is the discrete number of levels in the range $[-\pi,\pi)$.
	
	Therefore, by (\ref{eq:magresult}), applying Bayes' theorem as in (\ref{eq:otppt1angle1}) results in
	\begin{equation}
	p \left( \vec{m}_a [i] = x \middle| \vec{c}_a [i] = z \right)
	= \frac{p \left( \vec{m}_a [i] = x \right) \times \frac{1}{L}}
	{\frac{1}{L}}
	= p \left( \vec{m}_a [i] = x \right),
	\forall x \in M_a, \forall z \in C_a, \forall i \nonumber
	\end{equation}
\end{proof}

\begin{theorem}
	If the selection of $\vec{k} = (\vec{k}_m, \vec{k}_a)$ is purely random, then the VPSC is unconditionally secure (unbreakable).
\end{theorem}
\begin{proof}
	According to Claude Shannon's work in
	\cite{shannon49:secrecy}, given a plaintext message $m_1$ and cryptogram $c_1$,
	\begin{equation}
	\label{eq:otpshannon49}
	p \left( C = c_1 \middle| M = m_1 \right)
	= p \left( C = c_1 \right),
	\end{equation}
	
	Therefore, the VPSC is unconditionally secure if
	\begin{equation}
	\label{eq:otppt1}
	p \left( \vec{c} [i] = z \middle| \vec{m} [i] = x \right)
	= p \left( \vec{c} [i] = z \right),
	\forall x \in M, \forall z \in C, \forall i
	\end{equation}
	In other words, the VPSC's encrypted channel must provide
	no equivocation. No amount of cryptograms $\vec{c}_m [i]$ may provide
	any information about the original plaintext $\vec{m} [i]$.
	By Bayes' theorem,~\eqref{eq:otppt1} is equivalent to
	\begin{equation}
	\label{eq:otppt1eq}
	p \left( \vec{m} [i] = x \middle| \vec{c} [i] = z \right)
	= p \left( \vec{m} [i] = x \right),
	\forall x \in M, \forall z \in C, \forall i
	\end{equation}
	
	By expanding the magnitude and angle components of $x \in M, z \in C$,~\eqref{eq:otppt1eq} is equivalent to
	\begin{subequations}
		\begin{align}
		\label{eq:optangmag1}
		p \left( \vec{m} [i] = (x_m, x_a) \middle| \vec{c} [i] = (z_m, z_a) \right)
		= p \left( \vec{m} [i] = (x_m, x_a) \right),
		\forall x_m, z_m \in M, \forall x_a, z_a \in C, \forall i
		\end{align}
		
		In other words, the left-hand part of~\eqref{eq:optangmag1} can be expressed as
		\begin{align}
		\label{eq:optangmag2}
		p \left( \vec{m}_m [i] = x_m \land \vec{m}_a [i] = x_a \middle| \vec{c}_m [i] = z_m \land \vec{c}_a [i] = z_a \right)
		\end{align}
		
		For truly random $\vec{k}_m, \vec{k}_a$, the angle and magnitude components of both message and ciphertext are independent, therefore~\eqref{eq:optangmag2} is equivalent to 
		\begin{align}
		\label{eq:optangmag3}
		p \left( \vec{m}_m [i] = x_m \middle| \vec{c}_m [i] = z_m) \times p(\vec{m}_a [i] = x_a \middle| \vec{c}_a [i] = z_a \right)
		\end{align}
	\end{subequations}
	
	By lemmas~\ref{magequivocation} and~\ref{angequivocation}, and using~\eqref{eq:optangmag3},~\eqref{eq:optangmag1} reduces to
	\begin{equation}
	p \left( \vec{m} [i] = (x_m, x_a) \right)
	= p \left( \vec{m} [i] = (x_m, x_a) \right),
	\forall x_m, z_m \in M, \forall x_a, z_a \in C, \forall i
	\end{equation}
	
	Which is trivially true, therefore we have proven the equivalence in~\eqref{eq:otppt1}. This means that the encryption of the channel provides no equivocation, by Shannon's theorem of Theoretical Secrecy~\cite{shannon49:secrecy}. 
	
\end{proof}

Since the VPSC is unconditionally secure given truly random keys, no amount of cryptograms $\vec{c}[i]$ may provide
any information about the original plaintext $\vec{m}[i]$ for any
$i$. This gives an
encrypted signal $\vec{s}'$ the appearance of random noise (completely
random samples with zero correlation).

\subsection{The Influence of the Parameters on the Secrecy}

It is also important to note that public knowledge of $\phi$ or $\phi'$ does not
compromise the system; rather it is the secrecy of the key which protects the
message.

\begin{theorem}
	For any plaintext value $m$ of magnitude $m_m$, and for all $\phi = \varphi$ s.t. $\varphi_i \geq m_m$, if $m_m$ is encrypted using a purely random key $k_m$, the encrypted value $c$ provides no equivocation over $m$.
\end{theorem}
\begin{proof}
	We will prove this theorem by contradiction. Let us assume that there is some $\phi = \varphi_i$ s.t.
	\begin{equation}
	\label{eq:phiabs1}
	p \left( m = (x_m, x_a) \middle| c = (z_m, z_a) \right)
	\neq p \left( m = (x_m, x_a) \right)
	\end{equation}
	
	Following~\eqref{eq:optangmag3},~\eqref{eq:phiabs1} is equivalent to 
	\begin{equation}
	\label{eq:phiabs2}
	p \left( m_m = x_m \middle| c_m = z_m \right)
	\neq p \left( m_m = x_m \right)
	\end{equation}
	
	By expansion of the encryption operation, and since the \emph{a priori} probability of the key selection is independent from that of the plaintext
	\begin{equation}
	\label{eq:phiexpanded}
	p \left( m_m = x_m \land c_m = z_m \right)
	=  p \left( m_m = x \land k_m = (z-x) \bmod \varphi \right)
	\end{equation}
	
	However, that is in violation of theorem~\ref{magequivocation} which we have proven by~\eqref{eq:numeratorpt1} to be valid for all values of $\phi$ within the domain defined in the encryption function.
	
\end{proof}

Therefore, knowledge of the setting of $\phi$ may not compromise the system. This conforms with Kerckhoffs's principle since $\phi$ is part the cryptosystem and does not belong to the secret key.

\subsection{Pseudo Random Number Generator}
As proven, the VPSC is unconditionally secure. However,
this is only the case when using truly random keys.
Therefore, it is clear that the implementation of VPSC is only as
secure as the PRNG it uses for key generation.
This is true in general for all symmetric stream ciphers.
Since their encryption and decryption functions are simple, their strength
rests in the random number generator itself
\cite{marton10:randomness}.

There are many PRNGs which are considered cryptographically secure (CSPRNG) \cite{menezes1996handbook}.
These CSPRNGs are based on hash functions, block ciphers
\cite{paar10:crypto,petit08:cipher,bellare97:symmetric} and sometimes very
difficult mathematical problems such as elliptic curve generators
\cite{barker12:rng} or integer factorization like the Blum-Blum-Shub number generator
\cite{blum86:prng}. The usage of any of these types of CSPRNGs with the VPSC
is necessary in order to ensure that no signal correlations in $M$ will show through the
key-stream taken from $K$.

\section{Signal Synchronization}
\label{sec:sigsync}

In this Section, we propose a direct analytical method of synchronizing a VPSC decrypter to an encrypter at any arbitrary point in its key sequence.
A symmetric cryptographic system, such as the Vernam cipher, synchronizes in the following way. From the start of the transmission at time $t_0$, a deterministic key-stream is constantly generated at the source in order to encrypt the live signal. In parallel, the destination generates the same key-stream and uses it to decrypt the received signal at time $t_0+\rho$, where $\rho$ is the propagation delay.

In many real world applications, communication channels (such as wireless channels) broadcast signals continuously. In such instances, it is not practical to require the authorized decrypter to be ready to accept a transmission from start time  $t_0$. For example, in some communication networks, it is desirable to
allow authentic decrypters to be introduced into the network at an arbitrary
point in time $t_\theta$ where $t_\theta > t_0$.

These added decrypters, having no prior knowledge of the encrypter's status, cannot synchronize their key-stream under practical time constraints. This is especially true when the key generator has an exceptionally long cycle, as is the case with reliable CSPRNGs.
In order to generate a key-stream from an arbitrary point in a CSPRNG cycle, under practical time constraints, the generation of the targeted n\textsuperscript{th} key must be instantaneous.

For this reason, the usage of feedback ciphers is not practical since the generation of the n\textsuperscript{th} term requires generating all previous $n-1$ terms as well. In this Section, we will assume that the VPSC has been implemented with a cipher in counter (CTR) mode, such as AES-CTR \cite{dworkin01:ciphermodes}. This is usually done in the manner shown in \cite{hudde09:ciphers}. The counter mode, while being secure (if implemented correctly) and fully parallelizable (which is a performance advantage) \cite{ferguson10:engineering}, also offers the ability to access any value in the key stream independently from previous ones.

\subsection{VPSC Direct-Synchronization}

In order to initialize their cryptographic systems, both the encrypter and decrypter must use the same secret configuration. Let $D$ be the collection of all possible initial configurations for a system such that
\begin{equation}
\label{eq:config}
D = { (sc, st, g) | sc, st, g \in N}
\end{equation}
where $sc$ is the CSPRNG's initial counter (seed counter), $st$ is the start
time of the encryption relative to the encrypter's world-clock $t_{tx}$, and
$g$ is the key generation rate (values per time unit).
This initial configuration can be used by a decrypter to calculate the current CSPRNG-counter ($cc$) which can be used to generate the current key-frame $\vec{k}$.

However, in order to properly synchronize using time references, the decrypter must know at what time the transmitter encrypted each received frame (or what the system time was). In order to do this, the decrypter must take into account the propagation delay $\rho$ and the drift between its world-clock ($t_{rx}$) and the encrypter's ($t_{tx}$). This time delay $\varepsilon$ can be described as
\begin{equation}
\varepsilon = \left( t_{rx} - t_{tx} \right) + \rho
\end{equation}

Once $\varepsilon$ is known, the decrypter can calculate the current CSPRNG-counter value as follows
\begin{equation}
\label{eq:cc}
cc = \left[ \left( \left( t_{rx}+\varepsilon \right) - st \right) \times g \right] \bmod P
\end{equation}
where $P$ is the counter's period for the CSPRNG.
A special case arises when the selected CSPRNG generates, using a single counter, only a fraction of the random values needed to create a key-frame $\vec{k}$. In that case, we must round $cc$ to the nearest counter value which begins $\vec{k}$. This initial counter ($Ic$) can be calculated by
\begin{equation}
\label{eq:ic}
Ic = cc - cc \bmod (u)
\end{equation}
where $u$ is the number of counters needed to make a single key-frame.

\subsection{Time Delay Inference}
\label{sec:timedelayinfer}

Equations~\eqref{eq:cc} and~\eqref{eq:ic} and the encrypter's configuration $d \in D$, enable synchronization at any arbitrary point in time. However, in reality, a decrypter does not know $\varepsilon$ offhand. This is a problem since decryption of an encrypted stream at the wrong moment in time will result in a useless random signal (white noise).

Our proposed solution to this problem is to have the decrypter seek out $\varepsilon$ from the received signal. This can be achieved by first using~\eqref{eq:cc} without $\varepsilon$ to find a nearby counter. Afterwards, by finding the spot where the autocorrelation of the decrypted signal is the least like white noise, we can calculate $\varepsilon$. This works because each $\vec{k}$ will only decrypt its associated frame. Therefore, when a frame is found at a temporal distance from where we expected it, we can be sure that this offset is the $\varepsilon$ we are seeking. Moreover, if we average the results from this method several times with different keys, we can get better accuracy in the presence of noise, as shown in Fig.~\ref{fig:avgautocorrelations}.

Below are the steps for calculating $\varepsilon$, followed by a detailed explanation.

\begin{singlespace}
	\begin{enumerate}
		\item[] \textbf{Initialization}
		\item \label{itm:receive} Receive several frames-worth of encrypted samples into a buffer.
		\item \label{itm:keys} Using (\ref{eq:cc},~\ref{eq:ic}) without $\varepsilon$, calculate a number of different keys ($\vec{k}$) that span the middle of the buffer.
		\item \label{itm:window} Place a window of length $N$ at the beginning of the buffer.
		
		\item[] \textbf{Search for signal correlations}
		\item \label{itm:attemptdecrypt} Decrypt the window using each of the keys calculated in step~\ref{itm:keys}.
		\item Find the autocorrelations of the resulting frames.
		\item Calculate white noise comparison metrics for each of the autocorrelations.
		\item \label{itm:save} Save the metrics into different arrays (one for each key, respectively).
		\item \label{itm:shift} If the end of the buffer has not been reached, then shift the window by 1, sample and go to step~\ref{itm:attemptdecrypt}.
		
		\item[] \textbf{Calculate the time difference}
		\item \label{itm:avg} Superimpose and average the peaks found in each of the metric arrays (populated by step~\ref{itm:save}).
		\item \label{itm:getpeak} Calculate $\varepsilon$ by comparing the $t_{rx}$ used in step~\ref{itm:keys} to the highest peak value from step~\ref{itm:avg} on the temporal plane.
	\end{enumerate}
\end{singlespace}
\begin{figure}[h]
	\centering
	\includegraphics[width=\textwidth]{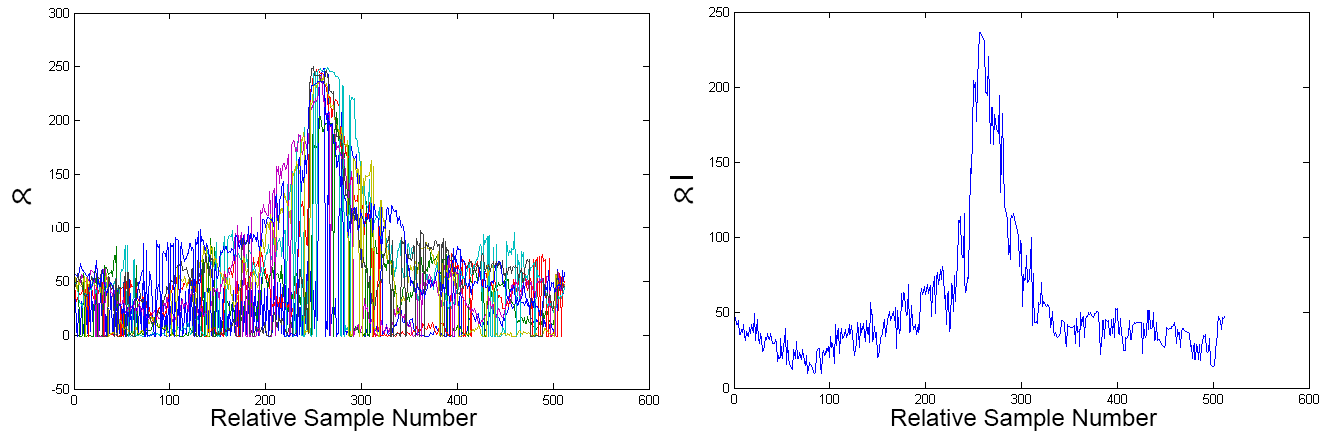} 
	\caption{Finding the start of a frame via signal correlations from decrypted frames. The frames were captured from a real dataset generated by a VPSC hardware prototype (see Section\ref{sec:simu}). The peaks indicate the window shift where the frame is the least like white noise (i.e. the actual start of a frame). Top: 8 metric arrays superimposed. Bottom: all 8 arrays averaged to mitigate noise.}%
	\label{fig:avgautocorrelations}
\end{figure}

In the initialization (steps~\ref{itm:receive}-\ref{itm:window}), a time span and all its details are captured in order to perform an offline analysis. At this point, we generate the keys for the frames that should fallout in the middle of the buffer by using~\eqref{eq:cc} without $\varepsilon$. Taking into account that $\varepsilon$ could be relatively large, the buffer should be significantly longer than the sum of the frames we are looking for.

Next is the search for signal correlations (steps~\ref{itm:attemptdecrypt}-\ref{itm:shift}). When a decrypted frame of samples has an autocorrelation that is not similar to that of white noise, then that is a good indication that there is some information gain. Of course, we assume that a communication channel has by definition some amount of information gain and therefore has low entropy \cite{gray11:entropy}. In order to measure the similarity of a decrypted signal to white noise we use the metric
\begin{equation}
\alpha =
\frac{\sum_{k=1}^{N-1} R_{\vec{d}, \vec{d}} [k]}
{R_{\vec{d}, \vec{d}} [0]}
\end{equation}
where $R_{\vec{x},\vec{x}} [n]$ is the discrete autocorrelation of the vector $\vec{x}$ at index $n$, and $\vec{d}$ is the decrypted frame. It can be observed that when $\vec{d}$ is in the form of white noise, the scalar $\alpha$ will yield a small value, since $\vec{d}$'s autocorrelation will be similar to that of the Dirac delta function. Similarly, as $\vec{d}$ becomes less similar to white noise, $\alpha$ will yield an increasingly larger value (indicating a dissimilarity to white noise).

Lastly, in steps~\ref{itm:avg} and~\ref{itm:getpeak}, we calculate the time difference. In each of the metric arrays there is a rise and fall in values. After superimposing and averaging the frames, the peak indicates where the signal holds the most information (i.e. the shift that covers the frame precisely). This is illustrated in Fig.\ref{fig:avgautocorrelations}. Since we have temporally found the frame for the respective key, we can now calculate $\varepsilon$. This can be done by solving
\begin{equation}
cc_{\vec{k}} = \left[ \left( \left( t_{usedRx_{\vec{k}}} - t_{peak_{\vec{k}}} \right) - st \right) \times g \right] \bmod P
\end{equation}
where $cc_{\vec{k}}$ and $t_{usedRx_{\vec{k}}}$ are the counter and ``current time'' used when creating $\vec{k}$ in step 2, and $t_{peak_{\vec{k}}}$ is the time sample which was derived from the head of the window (from that autocorrelation) for that shift.

It is important to note that the proposed method of synchronization is only practical if the decrypter's initial $t_{rx}$ has a time delay $\varepsilon$ less than some reasonable amount.

\section{Complexity \& Performance} \label{sec:performance}
Although it can be used in an offline manner, the VPSC's most useful application is in the real-time domain. Real-time applications are very sensitive to time delays. Therefore, it is important to discuss the VPSC's performance, in particular its processing delay.

Since the number of operations used at each end of the VPSC is relatively low,
analyzing its performance is relatively simple. As mentioned in
Section~\ref{sec:thevpsc}, the operations that take place at each end are:
$DFT$, $DFT^{-1}$, modulo-add or modulo-subtract and key generation.

Firstly, counter-mode ciphers are fully parallelizable \cite{ferguson10:engineering}. Therefore, in this Section, we will assume that
the CSPRNG calculations for generating the keys have been parallelized with the
encryption process, and are therefore not a factor in calculating the delay. Therefore, the
processing delay is dependent on all other operations mentioned above.

The most computationally expensive operations are the $DFT$ and $DFT^{-1}$ operations which are performed at each time a frame is encrypted or decrypted. The complexity of performing a $DFT$ or $DFT^{-1}$ on $N$ samples is $O(N log(N))$ \cite{lohne2017computational}. Since the VPSC operates on real-valued signals, a trick can be performed to reduce complexity. The trick is to put $N/2$ samples into each both real and complex inputs of a $\frac{N}{2}$ $DFT$ \cite{jones2010fast}. Therefore, the complexity of the encryption function $e_{\vec{k}}$ and decryption function $d_{\vec{k}}$ is
\begin{equation}
O\left(2 \cdot \frac{N}{2} \text{log}\left(\frac{N}{2}\right)\right) = O\left(N \text{log}\left(\frac{N}{2}\right)\right) 
\end{equation}

Although the $DFT$ and $DFT^{-1}$ are computationally costly operations, today it is
possible to find inexpensive hardware-accelerated DSP chips capable of
performing them at high speeds.
For example, the processor shown in \cite{analog13:ffttime} calculates a 1024-point
complex FFT at a 32-bit precision in $23.2 \mu s$. Using optimizations for real signals, it is possible to achieve even lower processing times. The same processor can also perform the
division required for the modulo operations in $8.75 ns$, which is negligible
compared to the FFT processing time.

Therefore, a frame of 1024-samples could undergo encryption and decryption in approximately 4 times the delay incurred by the calculation of one $DFT$
($2 \times DFT$ and $2 \times DFT^{-1}$). In the case of the chip mentioned
previously, this would be $92.8 \mu s$ (when using a complex input).
This delay is negligible compared to low-layer cryptosystems such as
those used in wireless LAN security \cite{hayajneh12:wlansec}.

Should an even shorter latency be necessary, it is possible to reprogram such chips
for a smaller frame size \cite{analog05:fftpoints} and get an even more significant decrease
in the FFT's computation time.

In some cases, there is no choice but to use slower hardware. Doing so means a compromise of either encrypting a smaller bandwidth (perhaps even smaller than the original signal) or to decrease the sampling rate, thereby decreasing the quality of the passing signal.
A solution to this problem could be the use of several VPSC encrypters and
decrypters in parallel. With such a configuration, each of the parallel devices would be in charge of processing
a particular band of the signal. Fig.~\ref{fig:parallelvpsc} illustrates an example of such a setup.

\begin{figure}[h]
	\centering
	\includegraphics[width=\columnwidth]{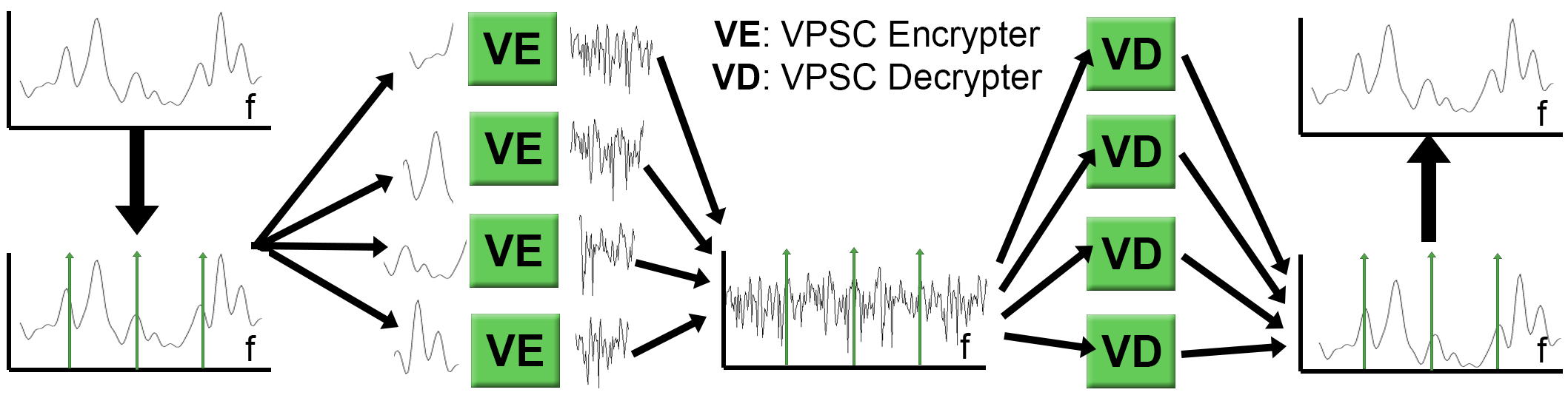}
	\caption{An illustration of how the VPSC can be used in parallel in order to encrypt an even wider bandwidth when using performance-limited hardware.}
	\label{fig:parallelvpsc}
\end{figure}

\section{Evaluation}
\label{sec:simu}
In Section~\ref{sec:cryptanalysis}, we proved that the VPSC is as secure as its CSPRNG, and theoretically unconditionally secure if a truly random key is used. However, the VPSC transforms a signal's harmonic composition, and the transformed signal may be sensitive to noise and channel interference. Therefore it is important to evaluate the VPSC's performance in practical scenarios, and not just in theory. 

In this Section, we verify the VPSC's practicality in realistic scenarios, such as both wired and wireless channels. We accomplish this by first evaluating the VPSC in a realistic channel simulator, and then by implementing the VPSC on actual hardware as a proof-of-concept.

\subsection{Simulation}
The most practical use for a physical signal cipher, is to protect channels which (1) are easily  accessible by an eavesdropper/attacker, and (2) do not require third parties to interpret and relay the signal (e.g., switches and routers). For these reasons, we evaluate the VPSC's performance in a wireless channel. 

We also note that the transference of a signal in a wireless channel is significantly more challenging than in wired channels. This is because wireless channels commonly suffer from multi-path propagation and other distortions. Therefore, evaluating the VPSC in such a channel demonstrates the VPSC's practicality and reliability in the real-world setting.

\subsubsection{Experiment Setup}
To evaluate VPSC in a wireless channel, we setup the simulation using the configuration of an LTE OFDMA mobile wireless channel. We evaluated the affect of noise and interference on the channel. Specifically, we experimented with additive Gaussian white noise (AWGN), multi-path propagation (MPP) interference, and the Doppler effect (assuming a mobile receiver). The MPP considered is based on a Rayleigh fading model which is considered as a realistic model within electrical engineering community. Rayleigh fading does not assume Line of Sight (LOS) path (in contrast to Rician fading) which is actually a very good approximation of the reality since we consider our scheme to be relevant for cellular Radio Access Network (RAN). LOS is generally usual in microwave backhaul which is typically point to point and not really relevant for our work. 

The default configurations taken for all simulations are available in Table \ref{tab:params}, unless mentioned otherwise.
Furthermore, in all simulations, we applied the VPSC's combined method of noise mitigation from Algorithms~\ref{alg:enc_combined} and~\ref{alg:dec_combined}, and set $\lambda = \left(N_0\right)^{-10}$, where $N_0$ represents the mean of the AWGN energy.

\begin{figure}[t] 
	\centering
	\captionof{table}{The parameters taken for the simulations.} 
	\label{tab:params} 
	\includegraphics[width=0.9\textwidth]{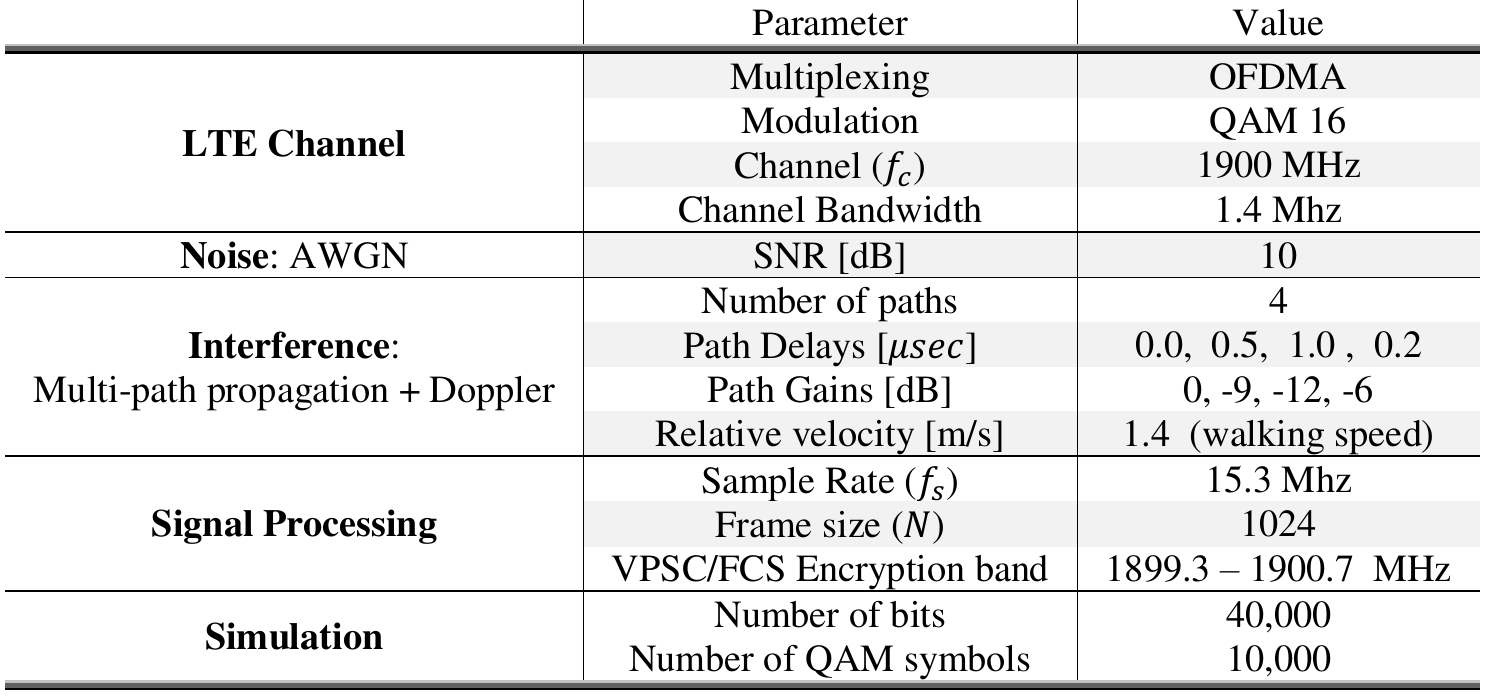}
\end{figure}

As a baseline comparison, we evaluated the VPSC to an unencrypted channel, and three other physical signal ciphers: FCS, ALM, and RSA (a description of these algorithms can be found in Section~\ref{sec:relatedwork}). For the VPSC and FCS, we encrypted the channel's bandwidth only, (tho other methods must encrypt the entire spectrum). For each simulation, we performed the following steps:

\vspace{1em}
\noindent\textbf{Transmitter}
\begin{enumerate}
	\item Generate four random bits.
	\item Modulate the four bits into a 16-QAM symbol.
	\item Generate the time-signal on carrier $f_c$ based on the 16-QAM symbol for $66.7\mu sec$ (the LTE OFDMA symbol duration).
	\item Encrypt the signal using one of the physical signal ciphers.
\end{enumerate}
\noindent\textbf{Channel}
\begin{enumerate}
	\item[5] Transmit the signal over a channel with AWGN and/or MPP. 
	\item[6] Buffering of the signal is applied to perform MPP and the Doppler effect.
\end{enumerate}
\noindent\textbf{Receiver}
\begin{enumerate}
	\item[7] Decrypt the received signal.
	\item[8] Demodulate the 16-QAM symbol from the decrypted signal using a \textit{hard decision} boundary.
	\item[9] Update the BER.
	\item[10] Return to step 1 unless 10,000 symbols have been received.
	
\end{enumerate}

\subsubsection{Experiment Results}
In Fig.~\ref{fig:BERgw_mpp}, we present the BER plots for all methods of encryption after passing through an AWGN channel, and a channel with both AWGN and MPP. The figure shows that the VPSC is robust to AWGN up to about 8dB SNR. Moreover, the VPSC is more robust than all other methods in the case of MPP, making the VPSC a better choice for wireless channels. We note that the VPSC performs better than an unencrypted signal up to 6dB SNR. This is because the combined noise mitigation technique adds some energy to the signal.  

We also note that the FCS is robust to AWGN, but not to MPP. This is because FCS shifts the majority of the signal's energy to the left or right of the carrier wave, thus increasing the interference of MPP and the Doppler effect across the channel's band. 

The ALM and RSA fail completely even in the presence of a minute amount of noise (24 dB SNR). Only at 300 dB SNR (i.e., \textit{negligible} noise) do these ciphers achieve perfect decryption. The reason ALM is sensitive to noise is because of the logarithm operations which it performs to encrypt signal samples. This increases the impact of the noise exponentially during the decryption process. The RSA method fails in the presence of noise since RSA function acts like a hash map. Therefore, a slight change to the key (input sample) can alter the output significantly.

In Fig.~\ref{fig:constl}, we provide a 16-QAM constellation plot of the demodulated symbols, after traversing various channels, and after various encryptions. We do not plot the results of ALM and RSA since the results are not informative (see Fig.~\ref{fig:BERgw_mpp}).

In Fig.~\ref{fig:BERmpp}, we examine the effect of the propagation delay in a MPP channel. In this simulation, we scaled up the propagation delays of the signal paths (Table~\ref{tab:params}). As a result, interference increases when there are symbols that overlap (ISI) which can be seen in the plot. The figure shows that the VPSC is nearly as robust to MPP as the original signal (without encryption). We note that FCS never has a `peak' of interference, but rather always stays at a non-zero BER. This is because most of the signal's energy falls on a single frequency component, and the random shift of this component's location helps mitigate ISI.

\begin{figure}
	\centering
	\includegraphics[width=0.48\textwidth]{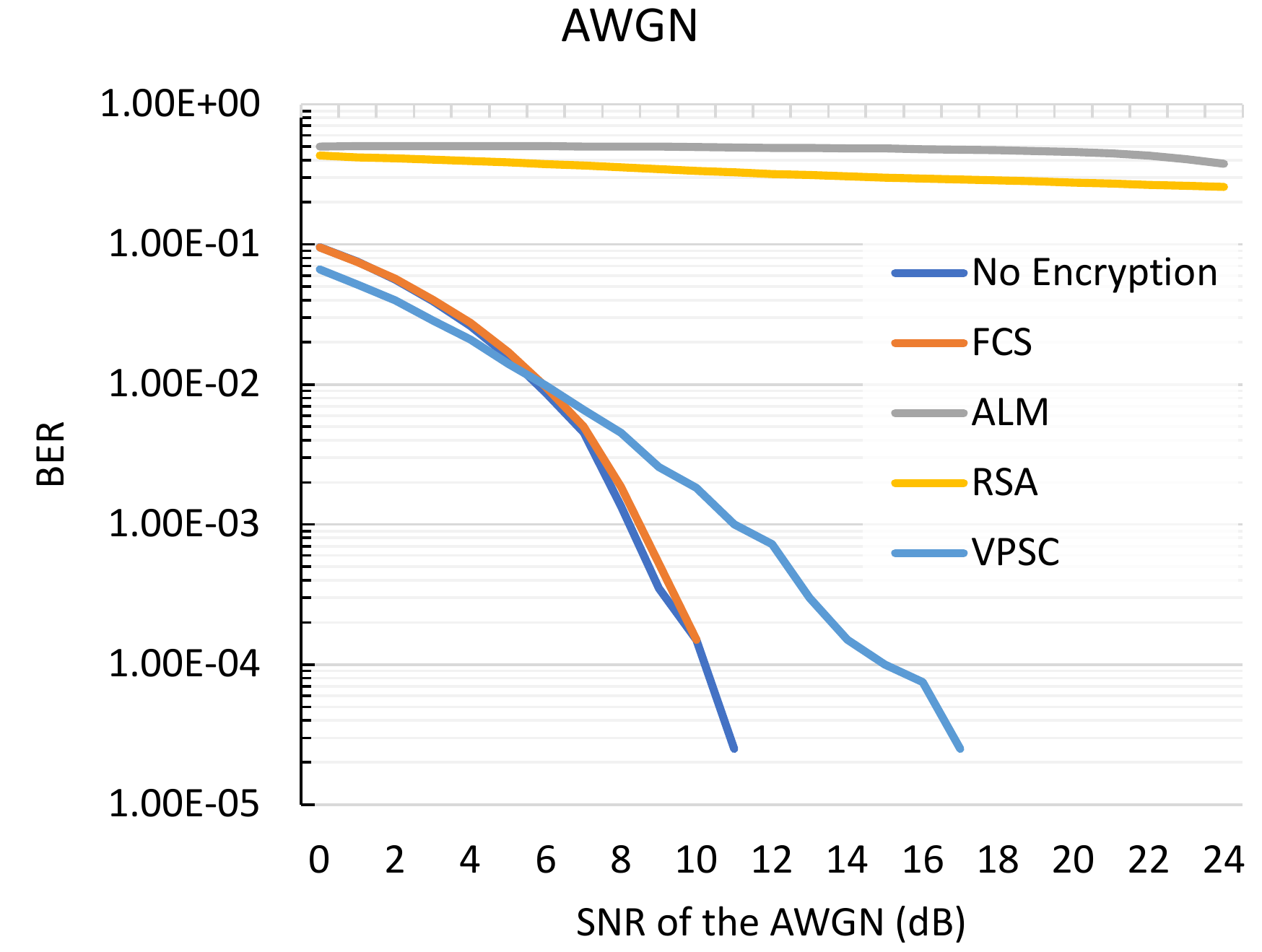}
	\includegraphics[width=0.48\textwidth]{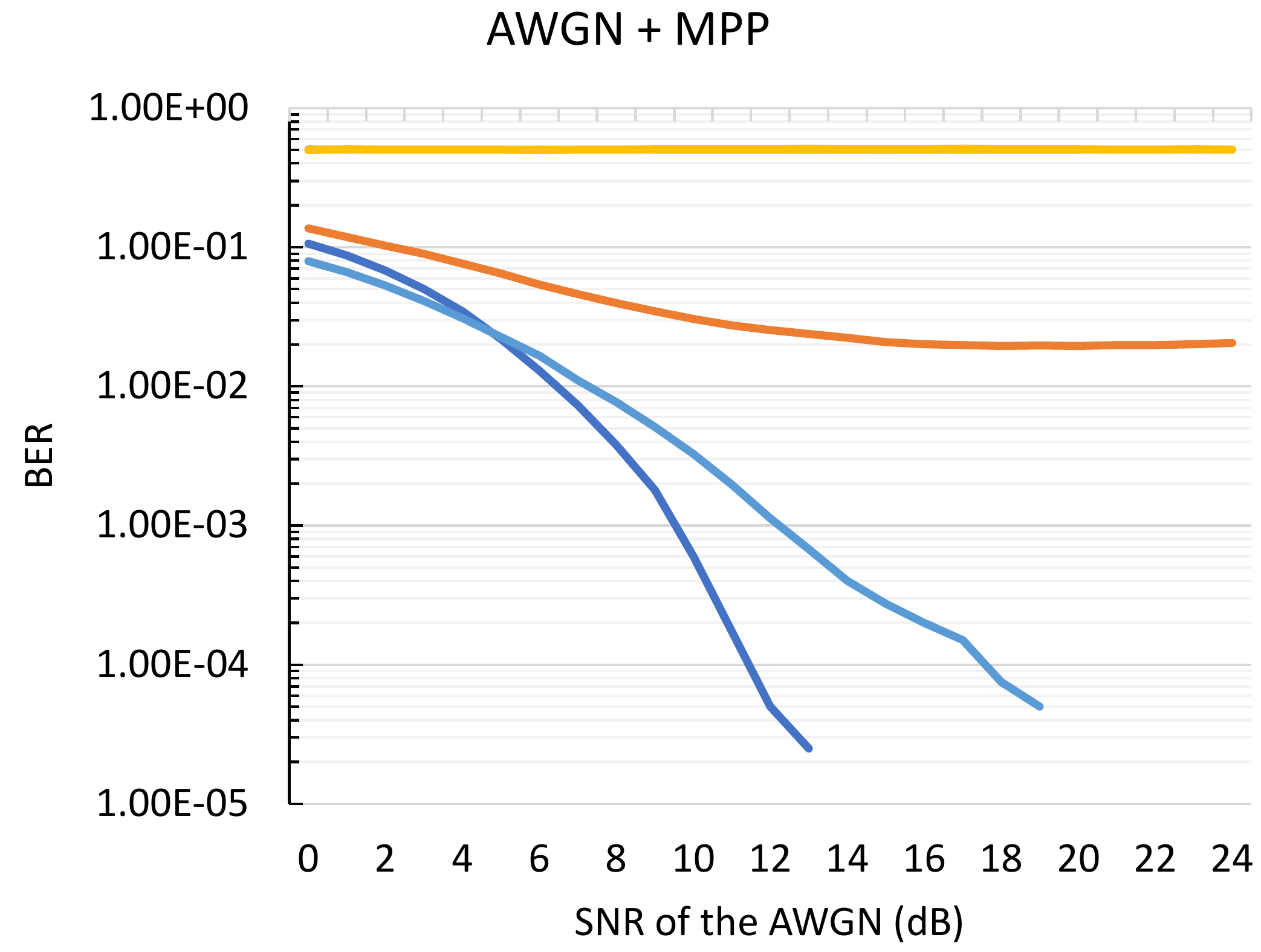}
	\caption{The bit error rate plots for all ciphers when introduced to AWGN (left) and both AWGN with multipath propagation and Doppler effect (right).}
	\label{fig:BERgw_mpp}
\end{figure}

\begin{figure}
	\centering
	\includegraphics[width=\textwidth]{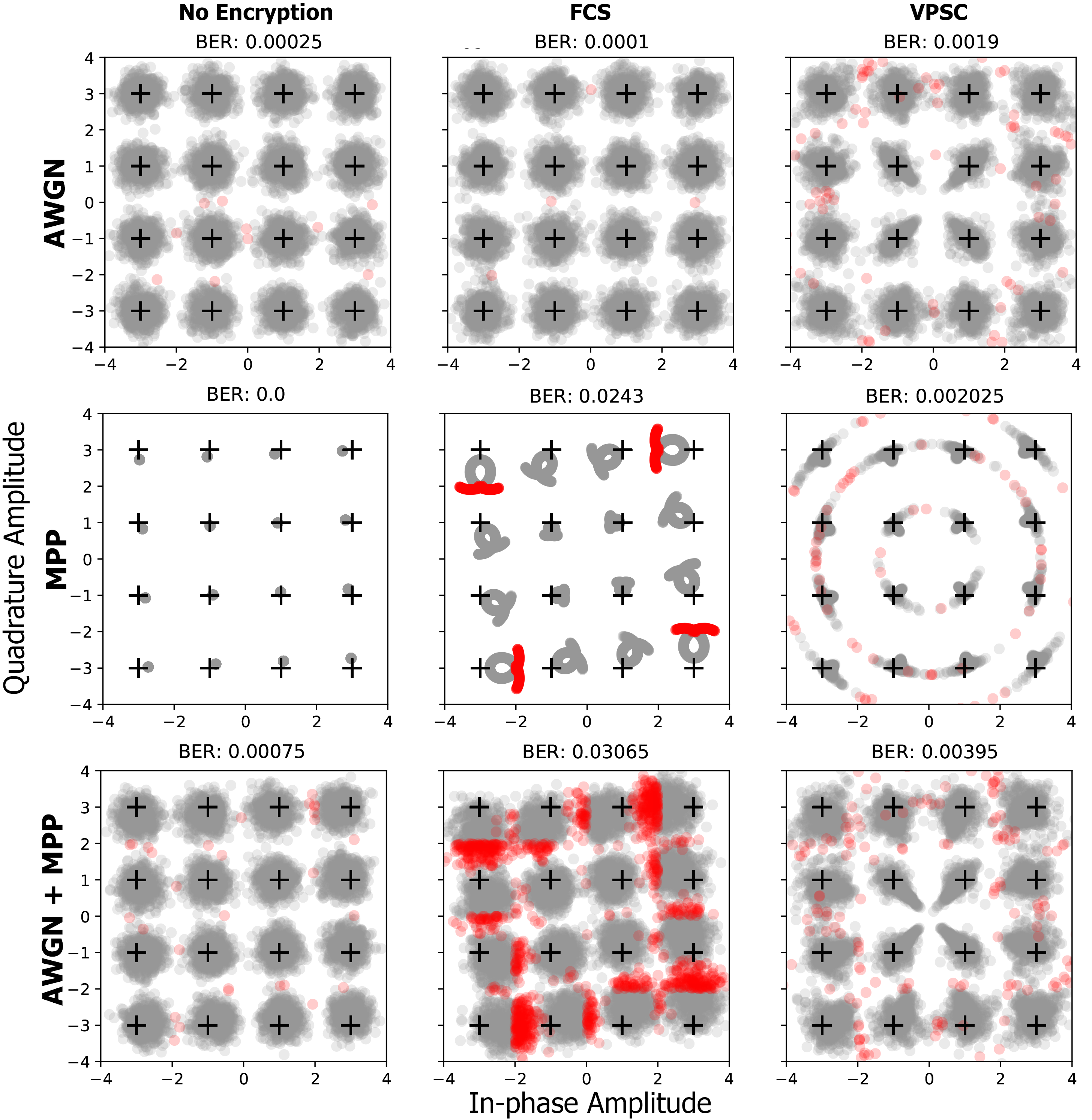}
	\caption{QAM-16 constellation plots showing the deciphered and demodulated symbols after passing through a 1900 MHz LTE OFDMA channel, with various types of noise and interference. The plots show the symbols which were correctly (gray) and incorrectly (red) demodulated. The `+' symbol marks the expected symbol locations.}
	\label{fig:constl}
\end{figure}

\begin{figure}[t]
	\centering
	\includegraphics[width=\textwidth]{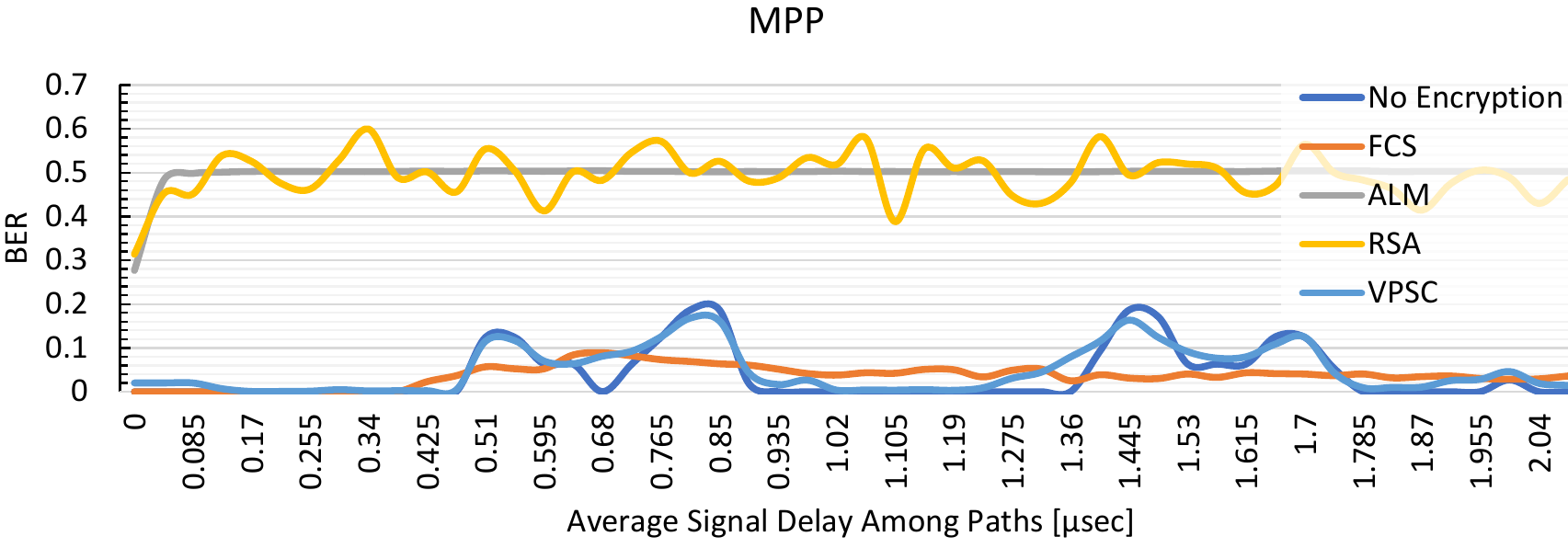}
	\caption{The bit error rate plots for all ciphers when introduced to multipath propagation and Doppler effect, with increasing propagation delays.}
	\label{fig:BERmpp}
\end{figure}

In order to measure the secrecy of the signal ciphers, we performed autocorrelations on the 16-QAM cryptograms (encrypted signals) generated from each of the ciphers. An autocorrelation measures the correlation of a signal with itself at various offsets. By measuring this self similarity, we can see evidence on whether or not a correlation attack may be performed. If a signal reveals no information about its contents, then the signal is essentially white noise. The autocorrelation of while noise is extremely low at every offset, except for when the signal overlaps itself completely.

In Fig.~\ref{fig:corr}, we present the autocorrelation of a 16-QAM signal, encrypted by each of the signal ciphers three times: each time with a different key. The plots reveal that the VPSC has the same autocorrelation as random noise. This makes sense since each of the FFT's frequency components holds a random magnitude and phase as a result of the selected key. The figure also shows that the FCS does not protect the channel's content, but rather is only obfuscates it. This is also apparent from the spectrum of the FCS's cryptograms, illustrated earlier in Fig.~\ref{fig:spectrum}. The ALM cipher provides a relatively good encryption since it resembles white noise. However, the ALM has \textit{spikes} in its autocorrelations. This means that some information is being leaked. This imperfection is likely due to the fact that ALM multiplies each sample with a value which is never zero. As a result, the cryptogram's distribution can reveal a portion of the contents. Given enough cryptograms, it may be possible to correlate out the ciphered signal. 

Finally, and surprisingly, the RSA method fails to protect the signal's contents. The figure shows that the original signal is significantly disclosed by the encrypted signal. This is due to an oversight in the physical signal cipher's implementation. Specifically, the RSA method encrypts the samples with the same private key. As a result, all samples with the same value are mapped to the same location after encryption. Although this process obfuscates the ciphered signal, some of its original frequencies are still retained. An illustration of this effect can be seen in Fig.\ref{fig:rsaFail}.

The reason the RSA method uses the same key is because it is computationally very expensive to generate (find large prime numbers) such a large amount of private keys and exchange the new public keys with the receiver.

\begin{figure}[t]
	\centering
	\includegraphics[width=\textwidth]{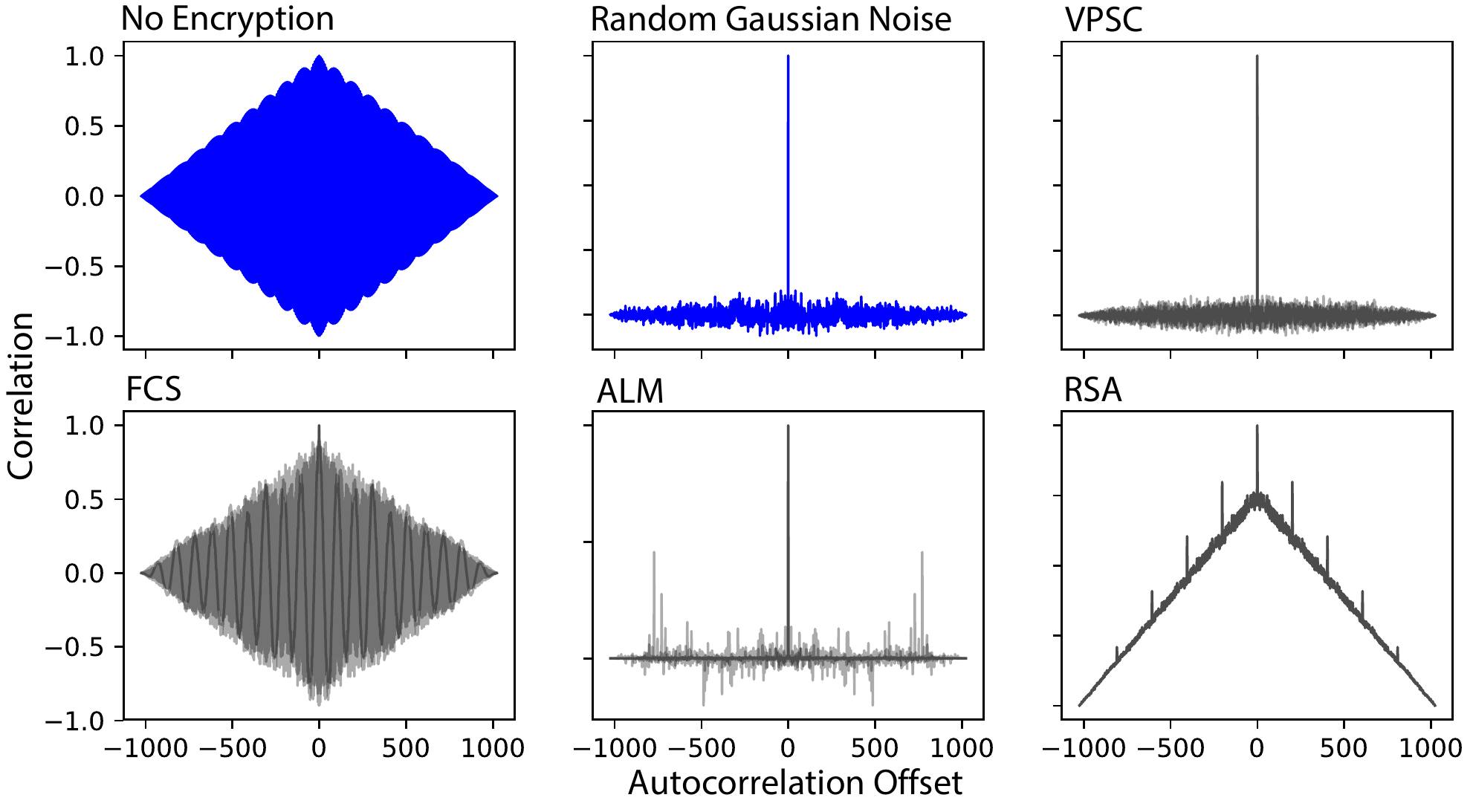}
	\caption{The autocorrelation of a 16-QAM signal, encrypted by each of the ciphers three times (each time with a different key).}
	\label{fig:corr}
\end{figure}
\begin{figure}[t]
	\centering
	\includegraphics[width=\textwidth]{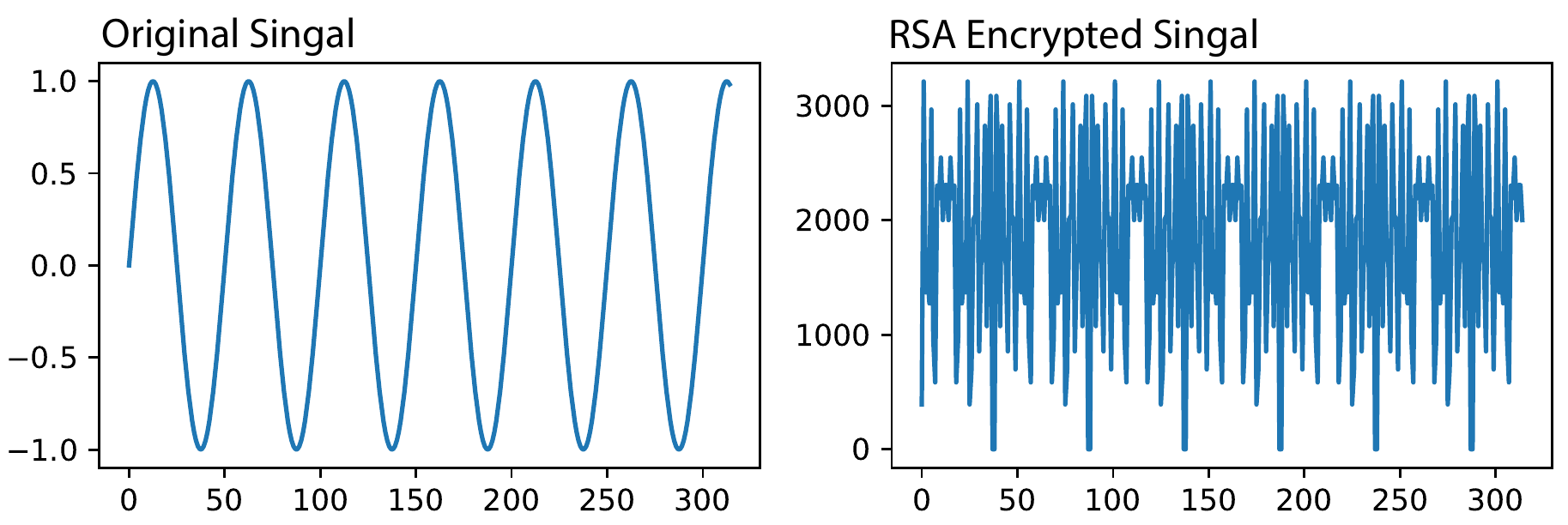}
	\caption{The RSA method's failure demonstrated by a sine wave on the left (plaintext) and the encrypted RSA signal on the right (ciphertext).}
	\label{fig:rsaFail}
\end{figure}

In summary, the VPSC is suitable for encrypting physical signals which traverse real-world channels. The VPSC is significantly more secure than other physical signal ciphers.

\subsection{Proof of Concept}
In order to fully observe and understand the VPSC, a simple prototype
was constructed and evaluated. The goal of the prototype was to provide a
proof of concept that would also illuminate any overlooked issues with
the VPSC's theory. Having said this, the processing speed (or amount of
processable bandwidth) of the devices used was not a critical part of
the evaluation.

The VPSC prototype was implemented across two Arduino Due development
boards. These boards embed 32-bit ARM core microcontrollers clocked at
84 MHz with 512 KB of flash memory and numerous embedded digital I/O pins.
One board was designated as the encrypter and the other as the decrypter.
A simple wire was used as the communications medium between the two boards.
The boards contain both digital-to-analog converters (DAC) and analog-to-digital converters (ADC), which were used for transmitting and receiving signals respectively.

In order to start the prototype, each board was given the same initial
configuration~\eqref{eq:config}. The time delay inference algorithm from Subsection~\ref{sec:timedelayinfer} was verified in MATLAB by using captured samples from the boards (Fig.~\ref{fig:avgautocorrelations}). With these samples, the algorithm correctly found and decrypted each frame.


The software developed to implement the VPSC was programmed in
object-oriented C++. Open-source libraries were used for the
Fast Fourier Transforms and the PRNG (SHA-1 as a counter-mode cipher).

\begin{figure}[t]
	\centering
	\includegraphics[scale=0.25]{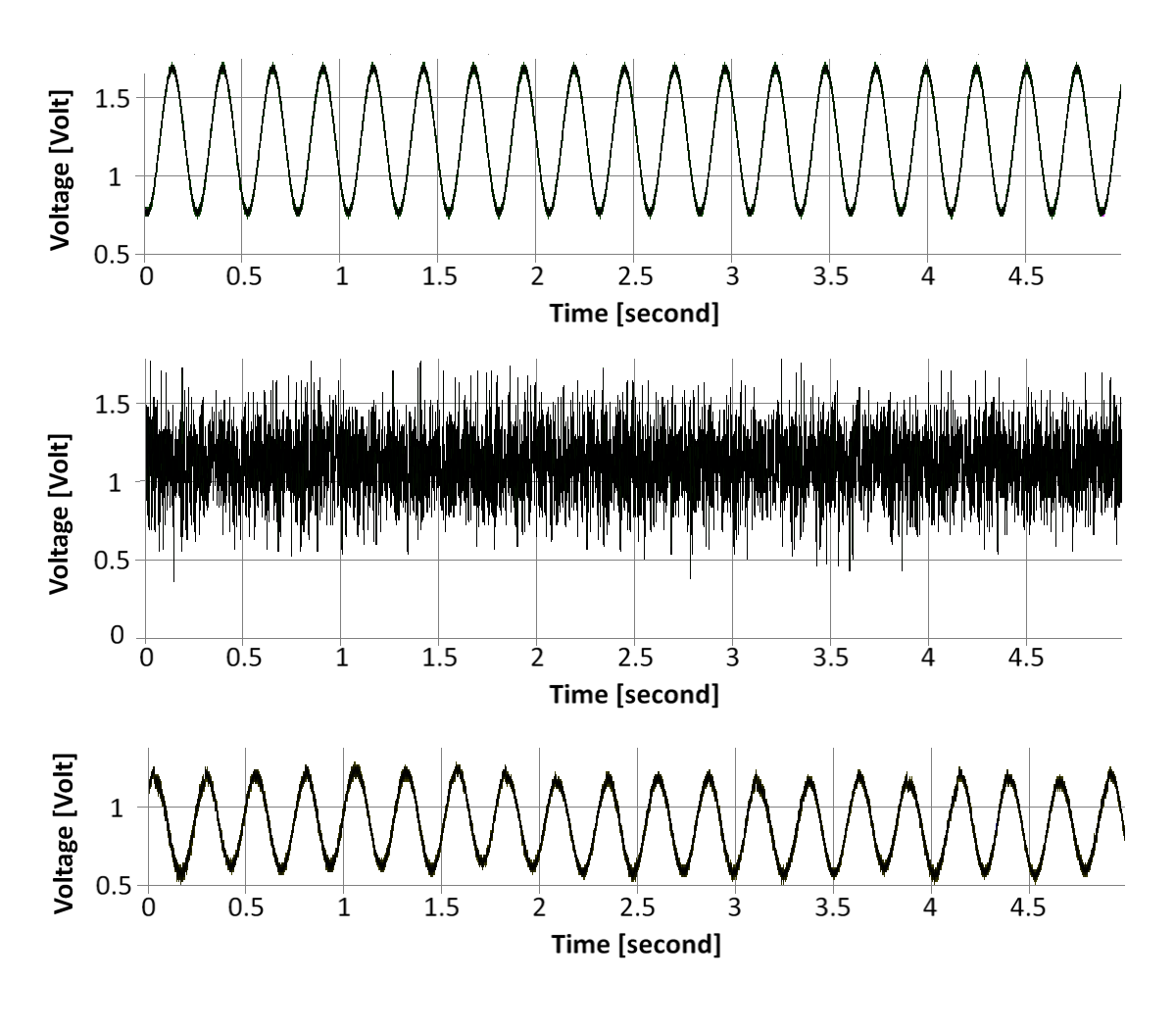} 
	\caption{The encryption and decryption of a 4 Hz sine wave. Top: the original signal $\vec{s}$ provided to the encryption board. Middle: The encrypted signal $\vec{s}'$ intercepted from the channel. Bottom: The decrypted signal delivered by the receiving board.}
	\label{fig:encdecsignal}
\end{figure}

As shown in Fig.~\ref{fig:encdecsignal}, the prototype was able to
successfully reconstruct a sine wave with an amplitude of 1 V and a frequency
of 4 Hz. The sample rate $f_s$ was 1 kHz and the frame size $N$ was 256 bits.
In this case, a sine wave was used as the source signal, with common
sine-based modulation schemes in mind.
By using superior hardware and a dedicated DSP chip, it is possible to process
larger bandwidths than the prototype's. This is particularly true due to
the computationally heavy Fourier transform purely implemented in software. This is discussed further in the
following section.  However, in our case the bottleneck was caused by the ADC/DAC components, which were only able to operate at a maximum sample rate of 2.2kHz.

\section{Discussion}
\label{sec:discussion}
This paper is a first look at the first physical signal cipher of its type. As with most new ideas, discussion can open up new avenues for further investigation and improvement. In this section we discuss some of the advantages and disadvantages of the VPSC. In addition, we also propose future work.

\subsubsection*{Key sharing}
A well-known problem with symmetric cryptographic systems is the issue of key sharing. For this reason, public-key (asymmetric) encryption has been largely favored in its place. Asymmetric encryption is the method where a public key is used by anyone to encrypt data, but only the holder of the private key can decrypt it. Symmetric cryptosystems such as the VPSC require that both communication parties have the same keys. It is clear that this can be an issue when initializing a connection over a public channel.

However, this issue may be avoided if the VPSC is converted into a hybrid cryptosystem; employing both symmetric and asymmetric cryptography. In this case the parties share a key (or initial seed) over a public channel using asymmetric cryptography. Afterwards, the key is used by both parties to initiate the VPSC. Of course, such a dialog would require the development of a protocol between two VPSC endpoints across a single link.

\subsubsection*{Security weaknesses}
When it comes to the time-tested Vernam cipher, there is one inherent weakness: the strength of the key generator.

For practical reasons, the prototype of the VPSC was implemented using a pseudo-random number generator (PRNG) as a means of symmetrically generating keys.
As discussed in Section~\ref{sec:cryptanalysis},
it is clear that the VPSC's security strength is completely dependent on the PRNG being used, because it directly inherits its weaknesses. Therefore, it is critical that a cryptographically secure PRNG (CSPRNG) be used, so that no statistical correlations show through the encrypted signal.

However, in order to use the direct-synchronization algorithm presented in Section~\ref{sec:sigsync}, the key generator must operate in counter (CTR) mode of operation. Since many CSPRNGs use feedback mechanisms in order to increase their randomness, it is difficult to determine the current key (one would have to compute $t$ steps, one for each key produced until now). Therefore, we recommend using a CSPRNG based on cryptographically secure hash functions (CSHF). This way, an attempt to find a sequential counter from a random number can be done in $O(1)$ time.

Another approach is to obtain the same key at both endpoints by sampling the channel noise \cite{chnlNoiseKey1,chnlNoiseKey2}. Implementation of this technique would be a practical option considering that the VPSC is already on \textit{Layer 1}. We leave this possibility to future work.

\subsubsection*{Channel detection}
It may be possible to use the VPSC to perform steganography (the hiding of a channel or information in plain sight). If the transmission power level is similar to that of common noise, then the transmission may go unnoticed.

An issue with the VPSC is that the frame length is detectable when analyzing the encrypted signal. This can be achieved by observing the frequency magnitudes over a period of time. However, this should not be a concern since the security of the signal relies on the information held by the frequencies and not the signal quality in which they were encrypted. A good rule of thumb when using the VPSC is to encrypt a band larger than the original signal's. Doing so will hide the details of the encrypted channels, much like a large curtain blocking a small object.

An alternative usage for the VPSC is as a signal jammer. If the encrypted band is wide and the transmission power is strong, it is possible to deploy an encrypted channel while all others are blocked by strong noise.

\subsubsection*{Bandwidth and wireless applications}
One of the major benefits of using the VPSC is its ability to retain the original signal bandwidth after encryption. Of course, due to hardware constraints, this limits the system to certain maximal bandwidths. One solution to this issue (as discussed in Section~\ref{sec:performance}) is to split the target band into segments, and then process them with independent parallel VPSC processors. We view this as a great advantage of the VPSC due to the fact that it operates on the frequency place.

The preservation of bandwidth makes the VPSC suitable for bandwidth-limited applications. Such applications include the wireless domain (including cellular and satellite communications) due to their shared communication medium. Spectrum in the world of wireless communications is an expensive commodity. Therefore any method for encrypting a signal without adding extra bandwidth is desirable.
In addition, encryption on higher layers sometimes results in overhead. This means that using the VPSC instead of a higher layer encryption scheme may in fact increase the amount of useful information (goodput) delivered by a channel.

\subsubsection*{Power}
One of the significant disadvantages with the VPSC is its energy consumption.
Due to the nature of the VPSC's encryption process, the encrypted signal's
frequencies each have an expected energy level of about $\phi/2$ watts.
Therefore, even a \textit{flat} signal with no power on any frequency will
have an increased power level once encrypted.
This means that when choosing to use the VPSC, one must consider the trade-off
between higher security and a lower transmission power level.

Adversely, this can be used as an advantage. If it is known in advance that
the source signal has a higher average frequency magnitude than $\phi/2$,
then the transmission power may be decreased.



\section{Conclusion}
\label{sec:conclusion}

In some cases there is a need to provide the highest level of security to
communications systems.
By using the VPSC, it is possible to encrypt any waveform signal to a high degree of secrecy while maintaining the same amount of bandwidth. Operation on the frequency domain has many advantages such as parallelization on the hardware level.
The level of security provided by the VPSC is
dependent on the strength of the PRNG which it uses (as with all one-time pads).
Moreover, we have proposed two techniques to ensure the stability of the system
when noise is introduced, and recommended using the combination of them both.

To evaluate the VPSC, we implemented three other known physical signal ciphers. We then simulated an LTE OFDMA wireless channel, with noise and interference, and measured the ciphers' performance in carrying a 16-QAM modulation. Furthermore, we explored the secrecy of each of the ciphers by examining the autocorrelations of their respective encrypted signals. Our evaluations demonstrated that the VPSC is not only the most suitable physical signaling cipher in noisy channels, but also the most secure.

The bottom line is that the VPSC offers a powerful method of encrypting any
waveform signal (as well as complex), with a trade-off between security and power efficiency.

\bibliographystyle{unsrt} 
\bibliography{IEEE_Forensics_Journal}

\begin{thebibliography}{10}

\bibitem{physectut}
Yi-Sheng Shiu, Shih-Yu Chang, Hsiao-Chun Wu, S.C.-H. Huang, and Hsiao-Hwa Chen.
\newblock Physical layer security in wireless networks: a tutorial.
\newblock {\em Wireless Communications, IEEE}, 18(2):66--74, April 2011.

\bibitem{EnergyEfJam}
Yee~Wei Law, Marimuthu Palaniswami, Lodewijk~Van Hoesel, Jeroen Doumen, Pieter
  Hartel, and Paul Havinga.
\newblock Energy-efficient link-layer jamming attacks against wireless sensor
  network mac protocols.
\newblock {\em ACM Trans. Sen. Netw.}, 5(1):6:1--6:38, February 2009.

\bibitem{phySecBook}
M.Bloch Bloch and J.~Barros.
\newblock {\em Physical-Layer Security: From Information Theory to Security
  Engineering}.
\newblock Cambridge University Press, 2011.

\bibitem{zhou2013physical}
X.~Zhou, L.~Song, and Y.~Zhang.
\newblock {\em Physical Layer Security in Wireless Communications}.
\newblock Wireless Networks and Mobile Communications. Taylor \& Francis, 2013.

\bibitem{nichols2002wireless}
R.K. Nichols and P.C. Lekkas.
\newblock {\em Wireless security: models, threats, and solutions}.
\newblock McGraw-Hill telecom professional. McGraw-Hill, 2002.

\bibitem{MACSec}
A.~Romanow.
\newblock Ieee standard for local and metropolitan area networks-media access
  control (mac) security.
\newblock {\em IEEE Std 802.1AE-2006}, pages 1--142, 2006.

\bibitem{shannon49:secrecy}
C.~E. Shannon.
\newblock Communication theory of secrecy systems.
\newblock {\em Bell system technical journal}, 28(4):656--715, 1949.

\bibitem{vernam19:otp}
G.~S. Vernam.
\newblock Secret signaling system, July 1919.
\newblock US Patent 1,310,719.

\bibitem{wyner1975wire}
Aaron~D Wyner.
\newblock The wire-tap channel.
\newblock {\em Bell System Technical Journal, The}, 54(8):1355--1387, 1975.

\bibitem{CsiszarKorner}
I~Csiszar and J.~Korner.
\newblock Broadcast channels with confidential messages.
\newblock {\em Information Theory, IEEE Transactions on}, 24(3):339--348, May
  1978.

\bibitem{menezes1996handbook}
A.J. Menezes, P.C. van Oorschot, and S.A. Vanstone.
\newblock {\em Handbook of Applied Cryptography}.
\newblock Discrete Mathematics and Its Applications. Taylor \& Francis, 1996.

\bibitem{SurveyOnKeyAlgs}
Harshala~B. Pethe and S.~R. Pande.
\newblock A survey on different secret key cryptographic algorithms.
\newblock {\em IBMRD's Journal of Management \& Research}, 3(1), 2014.

\bibitem{WiretapCapwithNoise}
H.~Yamamoto.
\newblock Rate-distortion theory for the shannon cipher system.
\newblock {\em Information Theory, IEEE Transactions on}, 43(3):827--835, May
  1997.

\bibitem{wiretapWithSecKey}
Wei Kang and Nan Liu.
\newblock Wiretap channel with shared key.
\newblock In {\em Information Theory Workshop (ITW), 2010 IEEE}, pages 1--5,
  Aug 2010.

\bibitem{articleLogscram}
M.I. Khalil.
\newblock Real-time encryption/decryption of audio signal.
\newblock 8:25--31, 02 2016.

\bibitem{jo10:crackingdsss}
Y.~Jo and D.~Wu.
\newblock On cracking direct-sequence spread-spectrum systems.
\newblock {\em Wireless Communications and Mobile Computing}, 10(7):986--1001,
  2010.

\bibitem{rivest1978method}
Ronald~L Rivest, Adi Shamir, and Leonard Adleman.
\newblock A method for obtaining digital signatures and public-key
  cryptosystems.
\newblock {\em Communications of the ACM}, 21(2):120--126, 1978.

\bibitem{audioFFTscramble}
A~Matsunaga, K.~Koga, and M.~Ohkawa.
\newblock An analog speech scrambling system using the fft technique with
  high-level security.
\newblock {\em Selected Areas in Communications, IEEE Journal on},
  7(4):540--547, May 1989.

\bibitem{symKeyEx2}
Paul Garrett and Daniel Lieman.
\newblock {\em Public-key Cryptography: Baltimore (Proceedings of Symposia in
  Applied Mathematics) (Proceedings of Symposia in Applied Mathematics)}.
\newblock American Mathematical Society, Boston, MA, USA, 2005.

\bibitem{symKeyEx1}
J.R. Vacca.
\newblock {\em Computer and Information Security Handbook}.
\newblock Elsevier Science, 2012.

\bibitem{marton10:randomness}
K.~Marton, A.~Suciu, and I.~Ignat.
\newblock Randomness in digital cryptography: A survey.
\newblock {\em ROMJIST}, 13(3):219--240, 2010.

\bibitem{paar10:crypto}
C.~Paar and J.~Pelzl.
\newblock {\em Understanding Cryptography: A Textbook for Students and
  Practitioners}.
\newblock Springer, 2010.

\bibitem{petit08:cipher}
C.~Petit, F.~X. Standaert, O.~Pereira, and M.~Yung T.~G.~Malkin.
\newblock A block cipher based prng secure against side-channel key recovery.
\newblock In {\em Proc. ASIACCS}, pages 56--65. ACM, 2008.

\bibitem{bellare97:symmetric}
M.~Bellare, A.~Desai, E.~Jokipii, and P.~Rogaway.
\newblock A concrete security treatment of symmetric encryption.
\newblock In {\em Proc. IEEE FOCS}, pages 394--403, 1997.

\bibitem{barker12:rng}
E.~B. Barker and J.~M. Kelsey.
\newblock {\em Recommendation for random number generation using deterministic
  random bit generators}.
\newblock NIST, 2012.
\newblock Special Publication 800-90A.

\bibitem{blum86:prng}
L.~Blum, M.~Blum, and M.~Shub.
\newblock A simple unpredictable pseudo-random number generator.
\newblock {\em Journal on Computing}, 15(2):364--383, 1986.

\bibitem{dworkin01:ciphermodes}
M.~Dworkin.
\newblock Recommendation for block cipher modes of operation-methods and
  techniques.
\newblock {\em NIST Special Publication 800-30A}, 2001.

\bibitem{hudde09:ciphers}
H.~C. Hudde.
\newblock Building stream ciphers from block ciphers and their security.
\newblock {\em Seminararbeit Ruhr-Universit\"{a}t Bochum}, 2009.

\bibitem{ferguson10:engineering}
N.~Ferguson, B.~Schneier, and T.~Kohno.
\newblock {\em Cryptography Engineering: Design Principles and Practical
  Applications}, chapter~4, page~70.
\newblock John Wiley \& Sons, 2012.

\bibitem{gray11:entropy}
R.~M. Gray.
\newblock {\em Entropy and information theory}.
\newblock Springer, 2011.

\bibitem{lohne2017computational}
Mathias Lohne.
\newblock The computational complexity of the fast fourier transform.
\newblock 2017.

\bibitem{jones2010fast}
Keith Jones.
\newblock Fast solutions to real-data discrete fourier transform.
\newblock In {\em The Regularized Fast Hartley Transform}, pages 15--25.
  Springer, 2010.

\bibitem{analog13:ffttime}
Sharc processor adsp-21367 reference.
\newblock
  http://www.analog.com/static/imported-files/data\textunderscore{}sheets/ADSP-21367\textunderscore21368\textunderscore21369.pdf,
  2013.
\newblock Datasheet.

\bibitem{hayajneh12:wlansec}
T.~Hayajneh, S.~Khasawneh, J.~Bassam, and A.~Itradat.
\newblock Analyzing the impact of security protocols on wireless lan with
  multimedia applications.
\newblock In {\em Proc. SECURWARE 2012}, pages 169--175, 2012.

\bibitem{analog05:fftpoints}
Parallel implementation of fixed-point ffts on tigersharc processor.
\newblock
  http://www.analog.com/static/imported-files/application\textunderscore{}notes/EE-263.pdf,
  2005.
\newblock EE-263.

\bibitem{chnlNoiseKey1}
David~M Wagner.
\newblock {\em Analysis of symmetric key establishment based on reciprocal
  channel quantization}.
\newblock PhD thesis, Rochester Institute of Technology, 2010.

\bibitem{chnlNoiseKey2}
Hadi Ahmadi and Reihaneh Safavi-Naini.
\newblock Secret keys from channel noise.
\newblock In {\em Advances in Cryptology--EUROCRYPT 2011}, pages 266--283.
  Springer, 2011.

\end{thebibliography}
%

\parpic{\includegraphics[width=1in,height=1.25in,clip,keepaspectratio]{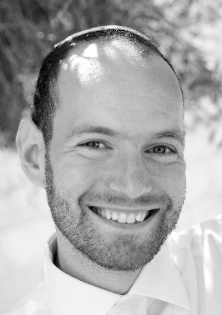}}
\noindent {\bf Yisroel Mirsky} is a post doctoral fellow in the Institute for Information Security \& Privacy at Georgia Tech (Georgia Institute of Technology). He received his PhD from Ben-Gurion University in 2018 where he is still affiliated as a security researcher. His main research interests include online anomaly detection, adversarial machine learning, isolated network security, and blockchain. Yisroel has published his research in some of the best cyber security conferences: USENIX, NDSS, Euro S\&P, Black Hat, DEF CON, CSF, AISec, etc. His research has also been featured in many well-known media outlets (Popular Science, Scientific American, Wired, Wall Street Journal, Forbes, BBC…). One of Yisroel's recent publications exposed a vulnerability in the USA's 911 emergency services infrastructure. The research was shared with the US Department of Homeland Security and subsequently published in the Washington Post.  
\\
\parpic{\includegraphics[width=1in,height=1.25in,clip,keepaspectratio]{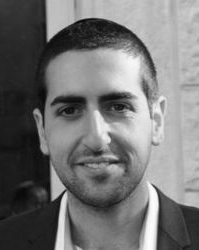}}
\noindent {\bf Benjamin Fedidat} received his B.Sc. in Communications from the Jerusalem College of Technology in 2013. He is now an M.Sc student in Computer Science at Bar-Ilan University (BIU) and a Systems programmer.\\
\vspace{3.5em}

\parpic{\includegraphics[width=1in,height=1.25in,clip,keepaspectratio]{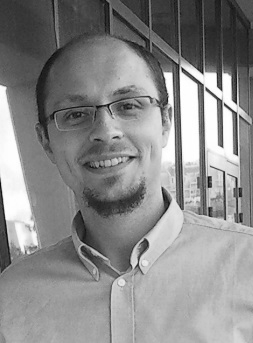}}
\noindent {\bf Dr. Yoram Haddad} received his B.Sc, Engineer diploma and M.Sc (Radiocommunications) from SUPELEC (leading engineering school in Paris, France) in 2004 and 2005, and his  PhD in computer science and networks from Telecom ParisTech in 2010. He was a Kreitman Post-Doctoral Fellow at Ben-Gurion University, Israel between in 2011-2012. He is currently a tenured senior lecturer at the Jerusalem College of Technology (JCT) in  Jerusalem, Israel. Yoram's main research interests are in the area of Wireless Networks and Algorithms for networks. He is especially interested in energy efficient wireless deployment, Femtocell, modeling of wireless networks, device-to-device communication, Wireless Software Defined Networks (SDN) and technologies toward 5G cellular networks.

\end{document}